\documentclass{amsart}
\usepackage{amsmath,amsthm}
\usepackage{amssymb}

\newtheorem{theorem}{Theorem}[section]
\newtheorem{corollary}{Corollary}

\newtheorem{lemma}[theorem]{Lemma}
\newtheorem{proposition}{Proposition}

\newtheorem*{problem}{Problem}
\theoremstyle{definition}
\newtheorem{definition}[theorem]{Definition}
\newtheorem{remark}{Remark}

\usepackage{hyperref}

\newcommand{\zm}{\mathbb Z / m \mathbb{Z}}
\newcommand{\zmb}{\left(\mathbb Z / m \mathbb{Z}\right)}
\newcommand{\weight}[1]{\mathrm{wt} \left( #1 \right) }
\newcommand{\lweight}[1]{\mathrm{wt}_{L} \left( #1 \right) }

\usepackage{tikz}
\usepackage{tikz-cd}

\tikzset{%
    symbol/.style={%
        ,draw=none
        ,every to/.append style={%
            edge node={node [sloped, allow upside down, auto=false]{$#1$}}}
    }
}
\usetikzlibrary{patterns}
\makeatletter
\tikzset{
        hatch distance/.store in=\hatchdistance,
        hatch distance=5pt,
        hatch thickness/.store in=\hatchthickness,
        hatch thickness=5pt
        }
\pgfdeclarepatternformonly[\hatchdistance,\hatchthickness]{north east hatch}
    {\pgfqpoint{-0.1pt}{-0.1pt}}
    {\pgfqpoint{\hatchdistance}{\hatchdistance}}
    {\pgfpoint{\hatchdistance-0.1pt}{\hatchdistance-0.1pt}}%
    {
        \pgfsetcolor{\tikz@pattern@color}
        \pgfsetlinewidth{\hatchthickness}
        \pgfpathmoveto{\pgfqpoint{0pt}{0pt}}
        \pgfpathlineto{\pgfqpoint{\hatchdistance}{\hatchdistance}}
        \pgfusepath{stroke}
    }
\makeatother

\usepackage{xcolor}

\usepackage{algorithm}
\usepackage[noend]{algpseudocode}

\DeclareMathOperator{\argmax}{argmax}
 
\newcommand{\vZ}{\mathbb{Z}}

\usepackage{algorithm}
\usepackage[noend]{algpseudocode}
\algdef{SE}[DOWHILE]{Do}{doWhile}{\algorithmicdo}[1]{\algorithmicwhile\ #1}%

\usepackage{mathtools}
\usepackage{enumerate}

\newcommand{\bs}{\mathbf{s}}
\newcommand{\ba}{\mathbf{a}}
\newcommand{\bb}{\mathbf{b}}
\newcommand{\bc}{\mathbf{c}}

\newcommand{\be}{\mathbf{e}}

\newcommand{\bz}{\mathbf{z}}

\newcommand{\by}{\mathbf{y}}
\newcommand{\bx}{\mathbf{x}}

\newcommand{\Id}{\mathrm{Id}}
\newcommand{\Z}{\mathbb{Z}}
\newcommand{\bA}{\mathbf{A}}

\newcommand{\bU}{\mathbf{U}}
\newcommand{\bzero}{\mathbf{0}}
\newcommand{\bB}{\mathbf{B}}

\newcommand{\bC}{\mathbf{C}}
\newcommand{\bP}{\mathbf{P}}

\newcommand{\bH}{\mathbf{H}}

\newcommand{\supp}[1]{\text{S}\left( #1 \right)}
\newcommand{\card}[1]{\left| #1 \right|}

\newcommand{\zps}[1][s]{\vZ/p^{#1}\vZ}
\newcommand{\zpsk}[2][s]{\left(\vZ/p^{#1}\vZ\right)^{#2}}

\newcommand{\concat}{%
  \mathbin{{+}\mspace{-8mu}{+}}%
}

\usepackage{nicematrix}
\usepackage{tikz,pgfplots}
\pgfplotsset{compat=newest}
\usepackage{subcaption}
\usetikzlibrary{calc, chains,
                fit,
                positioning,
                shapes}
\usepgfplotslibrary{fillbetween}

\usetikzlibrary{automata,arrows,positioning,calc}
\definecolor{mygray}{RGB}{211,211,211}

\usetikzlibrary{decorations.pathreplacing}

\title{On the Hardness of the Lee Syndrome Decoding Problem}

\author[V. Weger]{Violetta Weger}
\address{Department of Electrical and Computer Engineering, Technical University of Munich}
\email{violetta.weger@tum.de}

\author[K. Khathuria et al.]{Karan Khathuria}
\address{Institute of Computer Science, University of Tartu}
\email{karan.khathuria@ut.ee}

\author[]{Anna-Lena Horlemann}
\address{School of Computer Science, University of St. Gallen}
\email{anna-lena.horlemann@unisg.ch}

\author[]{Massimo Battaglioni}
\author[]{Paolo Santini}
\address{Department of Information Engineering, Marche Polytechnic University}
\email{m.battaglioni@staff.univpm.it}
\email{p.santini@staff.univpm.it}

\author[]{Edoardo Persichetti}
\address{Department of Mathematical Sciences, Florida Atlantic University}
\email{epersichetti@fau.edu}

\date{}

\begin{document}

\maketitle

\begin{abstract}
In this paper we study the hardness of the syndrome decoding problem over finite rings endowed with the Lee metric. We first prove that the decisional version of the problem is NP-complete, by a reduction from the $3$-dimensional matching problem. 
Then, we study the  complexity of solving the problem, by translating the best known solvers in the Hamming metric over finite fields to the Lee metric over finite rings, as well as proposing some novel solutions. 
For the analyzed algorithms, we assess the computational complexity in the asymptotic regime and compare it to the corresponding algorithms in the Hamming metric. 
\end{abstract}

\section{Introduction}

NP-complete problems play a fundamental role in cryptography, since systems based on these problems offer strong security arguments. 
In this work, we focus on the Syndrome Decoding Problem (SDP), that is, the problem of decoding an arbitrary linear code, which is the basis of code-based cryptography. 
This branch of public-key cryptography originated from the the seminal work of McEliece in 1978 \cite{McEliece1978} and Niederreiter in 1986 \cite{Niederreiter1986}, and is currently regarded as one of the most studied and consolidated areas in post-quantum cryptography \cite{NISTreport2016}.
SDP in the Hamming metric was proven to be NP-complete over the binary finite field in \cite{Berlekamp1978} and over arbitrary finite fields in \cite{barg1994some}.

In order to solve this problem, an adversary is forced to apply a generic decoding algorithm. Currently, the best known algorithms of this type, for the Hamming metric, are those in the Information Set Decoding (ISD) family. These exponential-time algorithms are very well studied and understood, and constitute an important tool to determine the size of the parameters that are needed to achieve a given security level (for an overview see \cite{meurer2013coding, ISDreview}).

In the last few years, in order to reduce the size of the public keys associated to code-based cryptosystems, there has been a growing interest in exploring metrics other than Hamming's; for example, in \cite{gaborit2016hardness} the hardness of the rank SDP and in \cite{sumrank} the hardness of the sum-rank SDP were studied. By exploring these new settings one might find a metric or an ambient space that is better suited for certain scenarios, in the sense that ISD algorithms are more costly than in the classical case, for a given size of input.
In this paper we study the hardness of decoding a random linear code over $\zm$ equipped with the Lee metric. 
This metric, first introduced in \cite{Lee1958,Ulrich1957}, has recently received  attention within the cryptographic community, for example in \cite{Horlemann2019} where codes defined over $\mathbb{Z} / 4 \mathbb{Z}$ in the Lee metric were considered. After the online version of this manuscript was posted, the authors of \cite{debris} published a related article, where they propose classical and quantum algorithms based on Wagner's approach to solve the Hamming and the Lee metric SDP over finite fields, thus restricting to $\mathbb{Z}/p\mathbb{Z}$ for the Lee metric. In the common scenarios, the results in \cite{debris} match with ours: ISD algorithms in the Lee metric  have a larger cost than their classical counterparts and thus we encourage further research of the Lee metric and its applications in cryptography. In particular, the Lee metric has the potential of providing shorter sizes in a zero-knowledge identification scheme such as Stern's \cite{Stern1994}, and thus, via the Fiat-Shamir transform, \cite{fiat} also in signature schemes. Apart from cryptographic purposes, devising a generic decoding algorithm for this metric is interesting per se, from a theoretical point of view.

In this paper we show how the reduction from \cite{barg1994some} also proves the NP-completeness of SDP over any finite ring endowed with an additive weight,  i.e., where the weight of a vector is given by the sum of the weights of its entries.
Thus, this proof is also valid for codes defined over $\zm$, equipped with the Lee metric. 
In addition, we  provide algorithms to solve SDP in the Lee metric  over  $\zps$, where $p$ is a prime and $s \in \mathbb{N}$ and perform a complexity analysis. In particular, we translate the ideas of algorithms such as Stern's ISD algorithm and the BJMM ISD algorithm to our setting but also propose new algorithms that are tailored to the particular structure of a parity-check matrix over $\zps$. For the asymptotic complexity analysis, we also require additional computations with respect to the results in \cite{saddle}, where only the case where $p$ is odd was treated. We observe that the ISD algorithms in the Lee metric have a larger cost than in the Hamming metric, which makes the Lee metric a promising alternative for the Hamming metric regarding cryptographic applications.

The paper is organized as follows: in Section \ref{sec:Notation} we introduce the notation used throughout the paper and   formulate some general properties of the Lee metric. 
The proof of the NP-completeness of the Lee metric version of the SDP and the Given  Weight Codeword Problem (GWCP) is given in Section \ref{sec:SDP}. In Section \ref{sec:ISD} we present information set decoding algorithms over 
$\mathbb{Z} / {p^s} \mathbb{Z}$ with respect to the Lee metric and assess their complexities. In Section \ref{sec:asymptotic} we compare the asymptotic complexity of all the Lee ISD algorithms and the Hamming ISD algorithms.
 Finally, we draw some concluding remarks in Section \ref{sec:concl}.

\section{Notation and Preliminaries}\label{sec:Notation}

In this section we provide the notation and the preliminaries used throughout the paper.

\subsection{Notation}

Let $p$ be a prime, $q$ a prime power and $m$ a positive integer. We denote with $\zm$ the ring of integers modulo $m$, and with $\mathbb{F}_q$ the finite field with $q$ elements, as usual. Given an integer $x$, we denote its absolute value as $\card{x}$.
The cardinality of a set $V$ is denoted as $\card{V}$ and its complement by $V^C$.
We use bold lower case (respectively, upper case) letters to denote vectors  (respectively, matrices). By abuse of notation, a tuple in a module over a ring, will still be denoted  by vector. The identity matrix of size $k$ is denoted by $\Id_k$; the $k\times n$ zero matrix is denoted as $\mathbf 0_{k\times n}$, while $\mathbf 0_n$ simply denotes the zero vector of length $n$. Given a vector $\mathbf{x}$ of length $n$ and a set $S \subseteq \{1, \ldots, n\}$,  we denote by $\mathbf{x}_S$ the projection of $\mathbf{x}$ to the coordinates indexed by $S$.
In the same way, $\mathbf{M}_{S}$ denotes the projection of the $k\times n$ matrix $\mathbf{M}$ to the columns indexed by $S$. 
The support of a vector $\mathbf{a}$  is defined as
$\supp{\ba} := \{j\hspace{2mm}|\hspace{2mm}a_j\neq 0\}$. 
 For $S \subseteq    \{1, \ldots, n \}$,  we denote by $\zmb^n (S)$  the vectors in $\zmb^n $ having support in $S$.

\subsection{Coding Theory in the Lee Metric}

In this section we recall the definitions and main properties of linear codes 
over finite rings endowed with the Lee metric, see also \cite{berlekamp1968algebraic}. 

\begin{definition}
For $x \in \zm$ we define the \emph{Lee weight} to be
\begin{equation*}
\lweight{x} := \min\{ x  ,  \mid m-x  \mid \},
\end{equation*}
Then, for  $\bx \in \zmb^n$, we define the \emph{Lee weight} to be the sum of the Lee weights of its coordinates,
\begin{equation*}
\lweight{\bx} :=  \sum_{i=1}^n \lweight{x_i}.
\end{equation*}
We define the \emph{Lee distance} of $\mathbf{x}$,  $\mathbf{y} \in \zmb^n$ as 
$$ \text{d}_L(\mathbf{x}, \mathbf{y} ) := \lweight{\mathbf{x}- \mathbf{y}} .$$ 
\end{definition}

\noindent The Lee sphere and the Lee ball of radius $w$ are denoted by
\begin{gather*} S_L(n,w,m) := \{ \mathbf{v} \in \zmb^n \mid \lweight{\mathbf{v}} = w \}, \\
B_L(n,w,m) := \{ \mathbf{v} \in \zmb^n \mid \lweight{\mathbf{v}} \leq w \},
\end{gather*}
respectively.
We denote their sizes by $F_L(n,w,m):=|S_L(n,w,m)|$ and $V_L(n,w,m):=|B_L(n,w,m)|$, respectively, and study their asymptotics in Subsection \ref{subsec:prelimLee}. 

\begin{definition} Let $m$ be a positive integer.
A \emph{ring-linear code} over $\zm$ of length $n$ is a $\mathbb{Z}/m\mathbb{Z}$-submodule of $\zmb^n$. 
\end{definition}

The \textit{type} of $\mathcal{C} \subseteq \left(\zm\right)^n$ is defined as
$$k= \log_m(\mid \mathcal{C}\mid)$$
and the \textit{rate} of $\mathcal{C}$ is then given by $R=k/n.$

The \emph{minimum Lee} \emph{distance} $d_L(\mathcal C)$ of a code $\mathcal C\subseteq \zmb^n$ is the minimum of all Lee distances of distinct codewords of $\mathcal C$:
$$d_L(\mathcal{C})= \min\{ d_L(x,y) \mid x \neq y \in \mathcal{C}\}.$$

Linear codes can be completely represented through a generator or a parity-check matrix.
\begin{definition}
A matrix $\mathbf G$ is called a \emph{generator matrix} for a (ring) linear code $\mathcal C$ if its row space  corresponds to $\mathcal C$. In addition, we call a matrix $\mathbf{H}$ a  \emph{parity-check matrix} for $\mathcal{C}$ if its kernel corresponds to $\mathcal{C}.$
\end{definition}
Note that such generator and parity-check matrices are not unique. If $m$ is not prime, even the number of rows of such matrices is not unique.

From Section \ref{sec:ISD} onward we will restrict ourselves to $\zps$, for some prime $p$ and positive integer $s$, since  $\zps$ provides the most studied case for ring-linear codes.

Using the fundamental theorem of finite abelian groups, we note that any code $\mathcal{C} \subseteq (\zps)^n$ is isomorphic to 
$$(\zps)^{k_1} \times (\mathbb{Z}/p^{s-1}\mathbb{Z})^{k_2} \times \cdots \times (\mathbb{Z}/p\mathbb{Z})^{k_s},$$ for some $k_1, \ldots, k_s \in \mathbb{N}$.
Thus, we can give the code additional parameters: $(k_1, \ldots, k_s)$ is called the \emph{subtype} of $\mathcal{C}$ and $K= \sum\limits_{i=1}^n k_i$ is called the \emph{rank}. Moreover, the parameter $k_1$ is called the \emph{free rank} of $\mathcal{C}$, and if $k_1 = k = K$, then the code is said to be \emph{free}.

\begin{proposition}[Systematic Form]\label{newgen}
Let $\mathcal{C}$ be a linear code in $\zpsk{n} $ of subtype $(k_1, \ldots, k_s)$. Then $ \mathcal{C}$ is permutation equivalent to a code having the following systematic
 parity-check matrix $\bH \in (\zps)^{(n-k_1) \times n}$,
\begin{equation}\label{systematicformH}
\bH =\begin{pmatrix}
\bA_{1,1} & \bA_{1,2} & \cdots & \bA_{1,s-1} & \bA_{1,s} & \Id_{n-K} \\
p\bA_{2,1} & p\bA_{2,2} & \cdots & p\bA_{2,s-1} & p\Id_{k_s} & \bzero_{k_s \times (n-K)}  \\
p^2\bA_{3,1} & p^2\bA_{3,2} & \cdots & p^2 \Id_{k_{s-1}} & \bzero_{k_{s-1} \times k_s}  & \bzero_{k_{s-1} \times (n-K)}  \\
\vdots & \vdots & & \vdots & \vdots & \vdots \\
p^{s-1}\bA_{s,1} & p^{s-1}\Id_{k_2} & \cdots &  \bzero_{k_2 \times k_{s-1}}  & \bzero_{k_2 \times k_s}  & \bzero_{k_2 \times (n-K)}
\end{pmatrix},
\end{equation}
where $\bA_{1,j} \in (\mathbb{Z} / p^s\mathbb{Z})^{(n-K) \times k_j}, \bA_{i,j} \in (\mathbb{Z} / p^{s+1-i} \mathbb{Z})^{k_{s-i+2} \times k_j}$ for $i > 1$.
\end{proposition}

In order to define information set decoders for the codes we consider, we need a suitable notion of information set, which is given next.
 \begin{definition}
 Consider a code $\mathcal{C}$ over
 $\zps$ of length $n$ and type $k$. 
 An \emph{information set} of $\mathcal C$ is a  set $I \subseteq \{1, \ldots, n\}$ of minimal size such that $$\mid \{\mathbf{c}_I \mid \mathbf c \in \mathcal{C} \} \mid
 = \mid \mathcal{C} \mid.$$ 
 \end{definition}
 Using parity-check matrices of the form  \eqref{systematicformH}, we can see that an information set for the respective code has cardinality $K$, corresponding to the first $K$ columns of $\bH$.
 
 Taking a random code over  $\zps$, it was shown in \cite{modules} that we have to consider the parity-check matrix in such a form, since the probability of having a non-free code is non-negligible.

\subsection{Asymptotics}\label{subsec:prelimLee}

In this section we provide the asymptotic analysis of the volume of the Lee ball (and Lee sphere) and the Gilbert-Varshamov bound. 
In the whole section, we assume that $q=p^s$ and define  $M:=  \left\lfloor \frac{q}{2} \right\rfloor$. Furthermore, we assume that all the code parameters $k=k(n)$, $K=K(n)$, $k_1=k_1(n)$ are functions of $n$, with the usual restrictions 
\begin{align*}
    0 \leq k_1(n) \leq k(n) \leq K(n) \leq n.
\end{align*} 
Further, let $t = t(n)$ be another function of $n$, which corresponds to the Lee weight of a vector of length $n$, hence it satisfies $0 \leq t(n) \leq Mn$. 

Since these quantities are bounded, we may define their relative limits in $n$
\begin{align*}
\lim\limits_{n \to \infty} \frac{k(n)}{n} &=: R,   \hspace{1cm}
\lim\limits_{n \to \infty} \frac{t(n)}{n} =: T,  \\
\lim\limits_{n \to \infty} \frac{k_1(n)}{n} &=: R_1, \hspace{0.7cm}
\lim\limits_{n \to \infty} \frac{K(n)}{n} =: R_I.
\end{align*}

The volume of the Lee balls (and spheres) can be approximated with the following results. We will use generating functions of the sizes of Lee spheres and Lee balls, which are known (see e.g. \cite{hastola}) to be $f(x)^n$ and $\frac{f(x)^n}{1-x}$, respectively, where $$f(x):=\left\{\begin{array}{ll} 1+2\sum_{i=1}^M x^i & \text{ if } q \text{ is odd}, \\ 1+2\sum_{i=1}^{M-1} x^i + x^M & \text{ if } q \text{ is even}. \end{array}\right.  $$ 
 We denote the coefficient of $x^t$ in a function $\Phi(x)$ by $[x^t]\Phi(x)$. 

We will make use of the following reformulation of a result from \cite{saddle} for estimating the size of the Lee balls and Lee spheres, for large $n$. 

\begin{lemma}\cite[Corollary 1]{saddle}\label{lemGardySole}
Let $\Phi (x) = f(x)^n g(x)$ with $f(0)\neq 0$, and $t(n)$ be a function in $n$. Set $T := \lim_{n\rightarrow \infty}t(n)/n$ and $\rho$ as the solution to 
$$ \underbrace{\frac{x f'(x)}{f(x)}}_{=:\Delta(x)} = T .$$
If $ \Delta'(\rho) >0$, and the modulus of any singularity of $g(x)$ is larger than $\rho$, then for large $n$
$$ \frac{1}{n} \log_q( [x^{t(n)}]\Phi(x)) \approx \log_q(f(\rho)) - T \log_q(\rho) + o(1) .$$
\end{lemma}
Using this, we get the following asymptotic behavior of Lee spheres and balls. 

\begin{lemma}\label{asympt_V} 
\begin{enumerate}
    \item If $q$ is odd and $0 \leq T<M$, then
    $$\lim\limits_{n \to \infty}  \frac{1}{n} \log_q( F_L(n, t(n), q))=\log_q(f(\rho)) - T \log_q(\rho),$$
    where $\rho$ is the unique real positive solution of $2 \sum_{i=1}^M (i-T) x^i = T$ and $$f(\rho)= 1 + 2\sum_{i=1}^M \rho^i = \frac{M(\rho+1)+1}{(1-\rho)(M-T)+1}.$$
    If moreover $T < \frac{M(M+1)}{2M+1}$, then also
    \[ \lim\limits_{n \to \infty} \frac{1}{n} \log_q(   V_L(n, t(n), q)  ) = \log_q(f(\rho)) - T \log_q(\rho), \]  
    
    \item If $q$ is even and $0 \leq T <M$, then
    \[ \lim\limits_{n \to \infty}  \frac{1}{n} \log_q( F_L(n, t(n), q)) = \log_q(g(\rho^\prime)) - T \log_q(\rho^\prime),\]
    where $\rho^\prime$ is the unique real positive solution of $$2 \sum_{i=1}^{M-1} (i-T) x^i + (M-T) x^M = T$$ and 
    \begin{align*}
        g(\rho^\prime) & = 1 + 2 \sum_{i=1}^{M-1} {\rho^\prime}^i + {\rho^\prime}^M  \\& = \frac{{\rho^\prime}^{M+1}(T-M) + {\rho^\prime}^M(T-M+1) + {\rho^\prime}(T-M) + T +M +1}{{\rho^\prime}(T-M) + M+1-T}.
    \end{align*} 
    If moreover $T < \frac{M}{2}$, then also 
    \[ \lim\limits_{n \to \infty}  \frac{1}{n} \log_q(   V_L(n, t(n), q)   ) = \log_q(g(\rho^\prime)) - T \log_q(\rho^\prime).\] 
\end{enumerate} 
\end{lemma}
\begin{proof}
\begin{enumerate}
    \item If $q$ is odd, with $f(z) := 1 + 2 \sum_{i=1}^M x^i$, we have that the generating functions for $F_L(n, t, q)$ and $  V_L(n,t,q)  $ are given by $f(x)^n$ and $\frac{f(x)^n}{1-x}$, respectively. The statement now follows from Lemma \ref{lemGardySole}, where the derivation of $f(\rho)$ can be found in the proof of \cite[Theorem 4]{saddle}. Notice that we do not need any restriction on $T$ in the first case, since here $g(x)=1$ does not have any singularities. In the second case the only singularity of $g(x)=\frac{1}{1-x}$ is $1$, hence we need that $\rho <1$. This is provided with the condition on $T$, since $\Delta(x) = \frac{2\sum_{i=1}^M ix^i}{1+2\sum_{i=1}^M x^i}$ is strictly increasing for $x\geq 1$ and 
    $$\Delta(1) = \frac{M(M+1)}{2M+1} > T = \Delta(\rho).$$
    
    \item If $q$ is even, with $g(z) := 1 + 2 \sum_{i=1}^{M-1} x^i+x^M$, we have again that the generating functions for $F_L(n, t, q)$ and $ V_L(n,t,q)  $ are given by $g(x)^n$ and $\frac{g(x)^n}{1-x}$, respectively, and that the statement follows from Lemma \ref{lemGardySole}, where the derivation of $g(\rho)$ can be found in the proof of \cite[Theorem 4]{saddle}. We again need that $\rho <1$ in the second case, because $\Delta(x) = \frac{2\sum_{i=1}^{M-1} ix^i + Mx^M}{1+2\sum_{i=1}^{M-1} x^i + x^M}$ is strictly increasing for $x\geq 1$ and 
    $$\Delta(1) = \frac{M}{2} > T = \Delta(\rho).$$
     
\end{enumerate} 
\end{proof}

\begin{remark}
Notice that asymptotically the size of the ball equals the size of the largest sphere inside the ball, as long as $T$ fulfills the prescribed conditions, which is approximately $T<M/2$ for both cases.

Note that if $q$ is odd and $T \geq M(M+1)/(2M+1)$ or if $q$ is even and $T \geq M/2$, we get that \( \lim\limits_{n \to \infty}  \frac{1}{n} \log_q(   V_L(n, t(n), q)  ) = 1.\) 
This can be easily observed as 
\begin{align*}
    1 &\geq  \lim\limits_{n \to \infty}  \frac{1}{n} \log_q(  V_L(n, t(n), q)  )   \\
    & \geq  \begin{cases}   \lim\limits_{n \to \infty}  \frac{1}{n} \log_q(   V_L(n, M(M+1)/(2M+1), q) ) & \mbox{if } q \mbox{ is odd} \\  \lim\limits_{n \to \infty}  \frac{1}{n} \log_q(   V_L(n, M/2, q) )  & \mbox{if } q \mbox{ is even} \end{cases} \\
    &  =1.
\end{align*}
The last equality follows from Lemma \ref{asympt_V}, where we get $\rho =1$ and $\rho^\prime = 1$, respectively. 
\end{remark}

Let $AL(n,d,q)$ denote the maximal cardinality of a code $\mathcal{C} \subseteq (\mathbb{Z}/q\mathbb{Z})^n$ of minimum Lee distance $d$ and let us  consider the maximal information rate 
$$R(n,d,q) := \frac{1}{n} \log_q(AL(n,d,q)),$$
for $0 \leq d \leq n M$. We define the relative minimum distance to be  $\delta := \frac{d}{nM}.$

\begin{theorem}[Asymptotic Gilbert-Varshamov Bound \cite{leeas}]\label{asympt_GV}
It holds that
$$\liminf\limits_{n \to \infty}R(n,\delta Mn, q) \geq \lim\limits_{n \to \infty} \left( 1 - \frac{1}{n} \log_q(  V_L(n, \delta Mn, q)  ) \right).$$
\end{theorem}

 In \cite[Theorem 22]{modules} it is shown that a code generated by a random matrix  achieves the Gilbert-Varshamov bound with high probability.

In ISD complexity analyses, we are interested in computing $\lim\limits_{n\to \infty} \frac{1}{n} \log_q( f(n))$, where $f(n)$ denotes the actual complexity function, all polynomial terms will become negligible. Note that a finite sum results in the asymptotics in a maximum of the limits of the summands, i.e., for a function $g(i)$ with $\lim\limits_{n \to \infty} 1/n \log_q(g(i)) = G(i)$, we have that
$$\lim\limits_{n \to \infty} 1/n \log_q \left(\sum_{i=1}^\ell g(i) \right) = \max\{G(i) \mid i \in \{1, \ldots, \ell \} \}.$$

\subsection{Uniform Distribution over $\mathbb{Z}/p^s\mathbb{Z}$ }
 
In the complexity analysis  we will assume that a random vector-matrix product is 
uniformly distributed. In fact, this is a key part for ISD algorithms, as we study the average-case complexity.
The following lemma gives the mathematical base for this assumption. 

\begin{lemma}\label{unif}
Consider a random vector $\bx \in \zpsk{K}$.
\begin{enumerate}
    \item If $\bx \in p\zpsk{K}$, then $p | \lweight{\bx}$.
    \item If $p | \lweight{\bx}$, then $\bx \in p\zpsk{K}$ with probability $\frac{F_L(K,\frac{v}{p},p^{s-1})}{F_L(K,v,p^s)}$.
    \item Let $\bA \in (\zps)^{K \times n}$ be chosen uniformly at random. Moreover, let $\bx \in (\zps)^{K}$ be chosen uniformly at random among the vectors that contain a unit, i.e., that do not live in $p\zpsk{K}$, possibly  with a weight restriction. Then $\bx\bA$ is uniformly distributed in $(\zps)^{n}$.
\end{enumerate}
\end{lemma}
\begin{proof}
\begin{enumerate}
    \item It follows from the definition of the Lee weight, that if $p|x_i$ for $i \in \{1, \ldots, K\}$, then $p|\lweight{x_i}$, which implies the statement.
    \item The number of $\bx \in p\zpsk{K}$ with $\lweight{\bx}=v$ is equal to the number of $\bar\bx \in (\mathbb Z/p^{s-1}\mathbb Z)^K$ with $\lweight{\bar\bx}=v/p$, which implies the statement. 
    \item Let $x_i$ be a unit entry of $\bx$. Then $\bx \bA = \left(\sum\limits_{j\neq i} x_j \bA_j\right) + x_i \bA_i$, where $\bA_i$ denotes the $i$-th row of $\bA$, and since $x_i \bA_i$ is uniformly distributed in $(\zps)^{n}$, so is $\bx \bA$.
\end{enumerate}
 
\end{proof}

Note that in the SDP it is assumed that $\bs$ is any element of $\zpsk{n-k_1}$. For ISD algorithms, however, it is assumed that $\bs$ is indeed a syndrome, that is, that there exists a solution $\be$ of weight $t$ such that $\bs = \be\bH^\top.$ 

In addition, ISD algorithms assume that the input $\bH$ is chosen uniformly at random, which 
by Lemma \ref{unif} implies that $\bs \in \zpsk{n-k_1}$ is uniformly distributed.

\section{The NP-Completeness of the General Syndrome Decoding Problem}\label{sec:SDP}

In this section we consider a more general definition of weight and prove that the Syndrome Decoding Problem (SDP) is NP-complete in this general setting, which includes the Lee weight setting as a special case. Note that all  problems in this section are easily seen to be in NP, since testing a solution clearly requires polynomial time. Hence showing NP-hardness automatically implies NP-completeness.

Let $R$ be a finite ring with identity, denoted by 1 and ${\rm wt}:R \to \mathbb{R}_{\geq 0}$ be a function on $R$ that satisfies the following properties: \begin{enumerate}
    \item   
    $\weight{0}=0$, 
    \item  $\weight{1}=1$ and $\weight{x} \geq 1 $ for all $x \neq 0$.  
\end{enumerate}
 
We call such a function a \textit{weight} (function). We extend this weight additively to $R^n$ by writing $\weight{\bx} = \sum_{i=1}^n \weight{x_i}$ for all $\bx = (x_1,\ldots,x_n) \in R^n$. In this section, we restrict ourselves to such \textit{additive weights}. 
Clearly, the Hamming weight and the Lee weight are examples of additive weights. 

Note that all results of this section hold also for  any weight fulfilling condition 1. and a more general version of condition 2., i.e., $\weight{x} \geq \lambda$ for all $x \neq 0$ and $\weight{1}=\lambda$ for some $\lambda \in \mathbb{R}_{> 0}$.   The argument is then similar, due to scaling.

We then define the syndrome decoding problem as:
\begin{problem}{\textbf{$(R,\rm wt)$-Syndrome Decoding Problem ($(R,\rm wt)$-SDP)}} \label{prob:AWSDP} \\
 Given $\bH \in R^{(n-k) \times n}$, $\bs \in R^{n-k}$ and $t \in \mathbb N$, is there a vector $\be \in R^n$ such that $\weight{\be} \leq t$ and $\be \bH^\top = \bs$?
\end{problem}

Berlekamp, McEliece and van Tilborg proved in \cite{Berlekamp1978} the NP-completeness of the syndrome decoding problem for the case of binary linear codes equipped with the Hamming metric (i.e., employing the Hamming weight in the problem definition).
 In particular, their proof was based on a reduction from the 3-dimensional matching (3DM) problem to SDP. 
In \cite{barg1994some}, Barg generalized this proof to an arbitrary alphabet size.

\begin{proposition}\label{prop:AW-SDP}
The $(R,\rm wt)$-SDP is NP-complete for any finite ring $R$ with identity and for any additive weight.  
\end{proposition}
As a corollary we obtain that also SDP in the Lee metric is NP-complete.

Since the proof of Proposition \ref{prop:AW-SDP} is very similar to the original proof of \cite{Berlekamp1978} and \cite{barg1994some}, we include it in the appendix for the sake of completeness. 
   
Let $R$ be a finite ring with identity and $k \leq n$ be positive integers. Let $\rm wt$ be an additive weight function on $R^n$ and  $k  \leq n$ be positive integers. We now define the given weight codeword problem as follows:
\begin{problem}{\textbf{$(R,\rm wt)$-Given Weight Codeword Problem ($(R,\rm wt)$-GWCP)}}\label{prob:GAWCP}
\\ Given $\bH\in R^{(n-k)\times n}$ and $w\in\mathbb{N}$, is there a vector $\mathbf c \in R^n$ such that $\weight{\mathbf c} =w$  and $ \mathbf c\mathbf H^\top = \mathbf 0_{n-k}$? 
\end{problem}

In order to prove the NP-completeness of of Problem \ref{prob:GAWCP}, we give a polynomial time reduction from the 3-dimensional matching problem.
\begin{problem}{\textbf{3-Dimensional Matching (3DM) Problem}}
\\Let $T$ be a finite set and $U \subseteq T \times T \times T$. Given $U,T$, decide if there exists a set $W \subseteq U$ such that $\card{W} = \card{T}$ and no two elements of $W$ agree in any coordinate. \end{problem}

\begin{proposition}\label{prop:GAWCP}
The $(R,\rm wt)$-GWCP is NP-complete for any finite ring $R$ with identity and for any additive weight.
\end{proposition}

The proof is similar to the original reduction of \cite{Berlekamp1978,barg1994some}; however, the adaptation is not trivial and therefore we report it next. 

\begin{proof}
We prove the NP-completeness by a reduction from the 3DM problem. To this end, we start with a random instance of 3DM, i.e., $T$ of size $t$, and $U \subseteq T \times T \times T$ of size $u$. Let us denote the elements in $T= \{b_1, \ldots, b_t\}$ and in $U= \{\ba_1, \ldots, \ba_u\}$. 
From this we build the matrix  $\overline{\bH}^\top \in R^{u \times 3t}$, like in \cite{Berlekamp1978}: 
\begin{itemize}
    \item for $ j \in \{1, \ldots, t\}$, we set $h_{i,j} = 1$ if $\ba_i[1]= b_j$ and $h_{i,j}=0$ else,
    \item for $ j \in \{t+1, \ldots, 2t\}$, we set $h_{i,j} = 1$ if $\ba_i[2]= b_j$ and $h_{i,j}=0$ else,
    \item for $ j \in \{2t+1, \ldots, 3t\}$, we set $h_{i,j} = 1$ if $\ba_i[3]= b_j$ and $h_{i,j}=0$ else.
\end{itemize}
 Let $M$ be the maximal weight of an element in $R$.
Then we construct $\bH^\top \in R^{(3tMu+3t+u) \times(3tMu+3t)}$ in the following way.
\begin{equation*}
    \bH^\top = \begin{pmatrix}
    \overline{\bH}^\top & \Id_u & \cdots & \Id_u \\
    -\Id_{3t} & \mathbf{0} & \cdots & \mathbf{0} \\
    \mathbf{0} & -\Id_u & & \mathbf{0} \\
    \vdots & & \ddots & \\
    \mathbf{0} & \mathbf{0} & & -\Id_u
    \end{pmatrix},
\end{equation*}
where we have repeated the size-$u$ identity matrix $3tM$ times in the first row.
Let us set $w= 3t^2M^2+4tM$ and assume that we can solve the $(R,\rm wt)$-GWCP on the instance given by $\bH, w$ in polynomial time. 
Let us consider two cases. 

\underline{Case 1:}
In the first case the $(R,\rm wt)$-GWCP solver returns as answer `yes', since there exists a $\mathbf c \in R^{3tMu +3t+u}$, of weight equal to $w$, such that $\mathbf c\bH^\top =\mathbf{0}_{3tMu+3t}$.
Let us write this $\mathbf c$ as
$$\mathbf c= (\overline{\mathbf c}, \mathbf c_0, \mathbf c_1, \ldots, \mathbf c_{3tM}),$$
where $\overline{\mathbf c} \in R^u, \mathbf c_0 \in R^{3t}$ and $\mathbf c_i \in R^u$ for all $i \in \{1, \ldots, 3tM\}.$
Then, $\mathbf c\bH^\top = \mathbf{0}_{3tMu+3t}$ gives the equations 
\begin{align*}
    \overline{\mathbf c}\overline{\bH}^\top - \mathbf c_0&= \mathbf{0}, \\
    \overline{\mathbf c} - \mathbf c_1 & = \bzero, \\
    & \vdots \\
    \overline{\mathbf c} - \mathbf c_{3tM} &= \bzero.
\end{align*}
Hence, we have that $\mathrm{wt}(\overline{\mathbf c}\overline{\bH}^\top)= \mathrm{wt}(\mathbf c_0)$ and 
$$\mathrm{wt}(\overline{\mathbf c}) = \mathrm{wt}(\mathbf c_1) = \cdots = \mathrm{wt}(\mathbf c_{3tM}).$$ 
Due to the coordinate-wise additivity of the weight, we have that 
$$\mathrm{wt}(\mathbf c) = \mathrm{wt}(\overline{\mathbf c}\overline{\bH}^\top) + (3tM+1)\mathrm{wt}(\overline{\mathbf c}).$$
Since $\mathrm{wt}(\overline{\mathbf c}\overline{\bH}^\top) \leq 3tM$, we have that  $\mathrm{wt}(\overline{\mathbf c}\overline{\bH}^\top)$ and $\mathrm{wt}(\overline{\mathbf c})$ are uniquely determined as the remainder and the quotient, respectively, of the division of $\mathrm{wt}(\mathbf c)$ by $3tM+1.$
In particular, if $\mathrm{wt}(\mathbf c) = 3t^2M^2+4tM,$ then we must have $\mathrm{wt}(\overline{\mathbf c})=tM$ and $\mathrm{wt}(\overline{\mathbf c}\overline{\bH}^\top)=3tM.$ 
Hence, the first $u$ parts 
of the found solution $\mathbf c$, i.e., $\overline{\mathbf c}$, give a matching for the 3DM in a similar way as in the proof of Proposition \ref{prop:AW-SDP}. For this we first observe that $\overline{\mathbf c}\overline{\bH}^\top$ is a full support vector and it plays the role of the syndrome, i.e., $\overline{\mathbf c} \overline{\bH}^\top = (x_1,\ldots,x_{3t})$, where $x_i \in R\setminus \{0\}$ having $\weight{x_i}=M$ for all $i \in \{1,\ldots,3t\}$. We note that $\overline{\mathbf c}$ has exactly $t$ non-zero entries, which corresponds to a solution of 3DM (for details see the proof of Proposition \ref{prop:AW-SDP} in the appendix).

\underline{Case 2:}
If the solver returns as answer `no', this is also the correct answer for the 3DM problem. In fact, it is easy to see that the above construction also associates any solution $W$ of the 3DM to a solution $\bc$ of the corresponding $(R, \rm wt)$-GWCP.

Thus, if such a polynomial time solver  for the $(R,\rm wt)$-GWCP exists, we can also solve the 3DM problem in polynomial time.   
\end{proof}

\begin{remark}
 We remark that the bounded version of this problem, i.e., deciding if a codeword $\bc$ with $\weight{\bc} \leq w$ exists, can be solved by applying the solver of Problem \ref{prob:GAWCP} at most $w$ many times.  
\end{remark}

\begin{remark}
 The computational versions of Problems \ref{prob:GAWCP} and \ref{prob:LSDP} are at least as hard as their decisional counterparts. 
Trivially, any operative procedure that returns a vector with the desired properties (when it exists) can be used as a direct solver for the above problems.
\end{remark}

As mentioned in the introduction, the previous results imply that the corresponding problems in the Lee metric are also NP-complete.

\begin{problem}{\textbf{Lee - Syndrome Decoding Problem (L-SDP)}}\label{prob:LSDP}
Let $m$ and $k\leq n$ be positive integers. Given $\bH\in (\mathbb{Z} / m \mathbb{Z})^{(n-k)\times n}$, $\bs\in (\mathbb{Z} / m \mathbb{Z})^{n-k}$ and $t\in\mathbb{N}$, is there a vector $\mathbf e \in (\mathbb Z / m \mathbb{Z})^n$ such that $\lweight{\be}\leq t$ and $ \mathbf e\mathbf{H}^\top  = \mathbf s$?
\end{problem}
 
 \begin{problem}{\textbf{Given Lee Weight Codeword Problem (GLWCP)}}\label{prob:GLWCP}
Let $m$ and $k\leq n$ be positive integers. Given $\bH\in (\mathbb{Z} / m \mathbb{Z})^{(n-k)\times n}$ and $w\in\mathbb{N}$, is there a vector $\mathbf c \in(\mathbb Z / m \mathbb{Z})^n$ such that $\lweight{\be} =w$ and $ \mathbf c\mathbf{H}^\top = \mathbf 0_{n-k}$?
\end{problem}

\begin{corollary}
The L-SDP and the GLWCP are NP-complete for any fixed $m \in \mathbb{N}$.
\end{corollary}

\section{information set Decoding in the Lee Metric}\label{sec:ISD}
  After showing the NP-completeness of the syndrome decoding problem in the Lee metric, we now aim at  assessing the complexity of solving it. For this, we adapt some well known ISD algorithms from the Hamming metric to the Lee metric, as well as derive new ISD algorithms that emerge from the special form of the parity-check matrix over a finite ring $\zps$. For each algorithm we provide the asymptotic analysis of the average cost, where the input parity-check matrix and syndrome are uniformly distributed over $\zps$.

The first ISD algorithm was proposed by Prange in 1962 \cite{Prange1962} and assumes that no errors happen in a randomly chosen information set. Although one iteration requires a low computational cost, the whole algorithm has a large overall cost, as many iterations have to be performed.
There have been many improvements upon the original algorithm by Prange, focusing on a more likely error pattern. Indeed, these approaches increase the cost of one iteration but, on average, require a smaller number of iterations (see \cite{Lee1988, Leon1988, stern, canteaut1998new, canteautsendrier, chabaud,  mmt, finiasz, bernstein2011smaller, Becker2012}). 
For an overview of the binary case see \cite{meurer2013coding, ISDreview}. With new cryptographic schemes proposed over general finite fields, most of these algorithms have been generalized to $\mathbb F_q$ (see \cite{ peters, niebuhr, klamti, hirose, interlando2018generalization}). 
 
In this paper we provide five different ISD algorithms for the ring $\zps$ equipped with the Lee metric. We start with the two-blocks algorithm, that may be seen as a generalization of Stern's algorithm in the Hamming metric \cite{stern}. Due to the special form of the parity-check matrix, dividing it in $s$-blocks, we also propose the $s$-blocks algorithm, which at the best of our knowledge has no known counterpart in the Hamming metric. 
In addition, we use the idea of Partial Gaussian Elimination (PGE) \cite{finiasz}, where one reduces the initial SDP to a smaller SDP instance. 
For this scenario, we provide both Wagner's approach \cite{wagner} of partitioning and the representation technique approach \cite{bjmm}, where we allow subvectors to overlap. Finally, following the idea of \cite{bjmm} we also introduce a combination of their approaches.
We notice that ISD algorithms in the Lee metric, for the special case of $\mathbb{Z}/ 4 \mathbb{Z}$, have already been studied in \cite{Horlemann2019}.


We have also considered generalized birthday decoding algorithms (GBA); however, we have found that they do not improve PGE algorithms. GBA is exploiting, as the name suggests, the birthday paradox; this allows in the Hamming metric to choose smaller list sizes, thus decreasing the cost and, at the same time, still finding a solution. However, the situation in the Lee metric is different: even if we consider the maximal list sizes, we run into the following problem: one has to ensure that at least one vector leading to a solution lives in the lists. This can be enforced by either repeating this step and thus employing this success probability or choosing a small enough number of positions on which one merges. In the Lee metric both solutions are not satisfactory, since we either have to choose the number of positions for merging to be 0, or the lowest cost is achieved by using the algorithm on just one level. 
Therefore, using even smaller lists would only worsen the situation and the cost. However, if one is interested in such an approach, it can be found for Wagner's algorithm over $\mathbb{Z}/p\mathbb{Z}$ in \cite{debris}.

\subsection{Two-Blocks Algorithm}\label{twoblocks}

In this section, we adapt  Stern's algorithm from the Hamming metric to the Lee metric, which also encompasses the basic Lee-Brickell's and Prange's algorithms as special cases. 

The idea of Stern's algorithm in the Hamming metric is to partition the chosen information set $I$ into two sets $X$ and $Y$ containing $v_1$ and $v_2$ errors, respectively. Moreover, it is assumed that there exists a zero-window $Z$ of size $z$ outside of the information set where no errors happen. Our adaptation is as follows.

Let us assume that the information set is $I=\{1, \ldots, K\}$, and that the zero-window is  $Z=   \{K+1, \ldots, K + z\}$; furthermore, let us define $J= \{1, \ldots, n\} \setminus (I \cup Z) = \{K+z+1, \ldots, n\}$. 
To bring the parity-check matrix $\bH \in (\zps)^{(n-k_1) \times n}$  into systematic form, we multiply by an invertible matrix $\bU \in (\zps)^{(n-k_1) \times (n-k_1)}$.
We can write the error vector partitioned into the information set part $I$, the zero-window part $Z$ and the remaining part $J$ as $\be = (\be_I, \mathbf{0}_z,  \be_J)$, with $\lweight{\be_I}=v_1+v_2$ and $\lweight{\be_J}=t-v_1-v_2$. 
We are in the following situation
\begin{align*}
\be\bH^\top\bU^\top &= \begin{pmatrix}
\be_I  & \bzero_z  & \be_J 
\end{pmatrix}  \begin{pmatrix}
\bA^\top  & \bB^\top & p\bC^\top \\ \Id_{z } & \bzero_{z \times (n-K-z)} & \bzero_{ z \times (K-k_1)} \\ \bzero_{(n-K-z) \times z} & \Id_{n-K-z} & \bzero_{(n-K- z) \times (K-k_1)}
\end{pmatrix} \\ & =  \begin{pmatrix} \bs_1 & \bs_2 & p \bs_3 \end{pmatrix} = \bs \bU^\top,
\end{align*}
where $\bA \in (\zps)^{z \times K}, \bB \in (\zps)^{(n-K-z) \times K}$, $\bC \in (\Z /p^{s-1} \Z)^{(K-k_1) \times K}$ and $\bs_1 \in (\zps)^{z}, \bs_2 \in (\zps)^{n-K-z}, \bs_3 \in (\Z /p^{s-1} \Z)^{K-k_1}$.
From this, we get the following three conditions
\begin{align}
\be_I\bA^\top &= \bs_1, \label{sternZqcond1} \\
\be_I\bB^\top + \be_J &= \bs_2, \label{sternZqcond2} \\
p\be_I \bC^\top &= p \bs_3. \label{sternZqcond3}
\end{align}
In the algorithm, we will define a set $P$ that contains all  vectors of the form  $\be_1 \bA^\top $ and $ \be_1\bC^\top$, and a set $Q$ that contains all vectors  of the form $\bs_1 - \be_2 \bA^\top$ and $\bs_3- \be_2 \bC^\top$. Whenever a vector in $P$ and a vector in $Q$ coincide, we call such a pair a \emph{collision}. 
 For each collision, we define $\be_J$ such that Condition \eqref{sternZqcond2} is satisfied, \emph{i.e.},
$$\be_J =  \bs_2 - \be_I\bB^\top.$$
If in addition $\lweight{\be_J}=t-v_1-v_2$, we have found the wanted error vector.
 The two-blocks algorithm in the Lee metric is depicted in Algorithm \ref{sternZq}.

\begin{algorithm}[h!]
\caption{Two-Blocks Algorithm over  $\zps$ in the Lee metric}\label{sternZq}
\begin{flushleft}
Input: $\bH \in (\zps)^{(n-k_1) \times n}$, $\bs \in (\zps)^{n-k_1}$, $t \in \mathbb{N}$, $K= m_1+m_2, z < n-K$ and $v_i < \min \{ \lfloor \frac{p^s}{2} \rfloor m_i, \lfloor \frac{t}{2} \rfloor\} $ for $i \in \{1,2\}$.\\ 
Output: $\be \in (\zps)^n$ with $\be\bH^\top=\bs $ and $\lweight{\be}=t$.
\end{flushleft}
\begin{algorithmic}[1]
\State Choose an information set $I\subset \{1, ...,n\}$  of size $K$ and choose a zero-window $Z \subset  I^C$ of size $z$, and define $J=  (I \cup Z)^C$.
\State Partition $I$ into $X$ of size $m_1$ and $Y$ of size $m_2=K-m_1$.
\State Compute $\bU\in (\zps)^{(n-k_1) \times (n-k_1)}$, such that \begin{align*} (\bU\bH)_{I}=  \begin{pmatrix} \bA \\ \bB \\ p\bC \end{pmatrix}, \   (\bU\bH)_{Z}=  \begin{pmatrix} \Id_z \\ \bzero_{(n-K-z) \times z} \\ \bzero_{(K-k_1) \times z} \end{pmatrix}, \ (\bU\bH)_J=  
	 \begin{pmatrix} \bzero_{z \times (n-K-z)} \\ \Id_{n-K-z} \\ \bzero_{(K-k_1) \times (n-K-z) } \end{pmatrix},  
\end{align*}	 where $\bA\in(\zps)^{z \times K }, \bB \in (\zps)^{(n-K-z) \times K}$ and $\bC \in (\Z/ p^{s-1}\Z)^{(K-k_1) \times K}$.
\State Compute $\bs\bU^\top = \begin{pmatrix}
\bs_1 & \bs_2 &p\bs_3
\end{pmatrix}$, where $\bs_1 \in (\zps)^z, \bs_2 \in (\zps)^{n-K-z}$ and $\bs_3 \in (\Z /p^{s-1} \Z)^{K-k_1}$.
\State Compute the set 
$$P= \{( \be_1\bA^\top, \be_1 \bC^\top, \be_1) \mid \be_1 \in (\zps)^{K}(X), \  \lweight{\be_1}=v_1\}.$$
\State Compute the set 
$$Q= \{( \bs_1-  \be_2\bA^\top, \bs_3 - \be_2 \bC^\top, \be_2) \mid \be_2 \in (\zps)^{K}(Y), \  \lweight{\be_2}=v_2\}.$$
\For{$(\ba, \bb, \be_1) \in P$} \For{$(\ba, \bb, \be_2) \in Q$}
		\If{$\lweight{\bs_2 - (\be_1+\be_2)\bB^\top}=t-v_1-v_2$}
		\State Return $\be, $ such that $\be_I = \be_1+ \be_2$, $\be_Z = \bzero_z$ and $\be_J =  \bs_2- (\be_1+\be_2)\bB^\top$.
				\EndIf
				\EndFor
				\EndFor
		\State  Start over with Step 1 and a new selection of $I$.
\end{algorithmic}
\end{algorithm}
For the following complexity analysis, we first recall the assumptions made in Section \ref{subsec:prelimLee}. Let all the code parameters $k(n), k_1(n), K(n), t(n)$ be functions of $n$, and define: 
\[R := \lim_{n \to \infty} \frac{k(n)}{n}, R_1 := \lim_{n \to \infty} \frac{k_1(n)}{n}, R_I := \lim_{n \to \infty} \frac{K(n)}{n}, T := \lim_{n \to \infty} \frac{t(n)}{n}.\]
We fix the real numbers $V,L$ with $0 \leq V \leq  T/2$ and $0 \leq L \leq  1-R$, such that 
$0 \leq T-2V \leq M(1-R-L).$
    Then we fix the internal algorithm parameters $m_1=m_2 = R_I/2$ and $v_1,v_2, \ell$ which we see as functions depending on $n$, such that
$ \lim\limits_{n \to \infty} \frac{v_i}{n}=V$ and $ \lim\limits_{n \to \infty} \frac{\ell}{n}=L$.

In order to ease the  asymptotic formulas, we introduce the following notation
\begin{align*}
    S(R,V) \coloneqq & \lim\limits_{n \to \infty} \frac{1}{n} \log_q(F_L(k,v,q)) = R \lim\limits_{k \to \infty} \frac{1}{k} \log_q(F_L(k,v,q)), \\
    L( R,V) \coloneqq &  \lim\limits_{n \to \infty} \frac{1}{n} \log_q(  V_L(k,v,q) ) = R \lim\limits_{k \to \infty} \frac{1}{k} \log_q(   V_L(k,v,q) ),
\end{align*}
for which we will use the results of Lemma \ref{asympt_V}.

\begin{theorem}\label{sternZqcost}
Algorithm  \ref{sternZq} over $\zps$ equipped with the Lee metric has the following asymptotic average time complexity:
\begin{gather*}     -2S(R_I/2, V)- S(1-R_I-L, T-2V) + S(1,T)+ \\
    \max\{S(R_I/2, V), 2S(R_I/2,V)-L-R_I+R_1\}.
\end{gather*}
\end{theorem}

\begin{proof}
The only steps that contribute to the asymptotic complexity, are the construction of $P$ and $Q$, finding the collisions between the two sets and the number of iterations required. 

The construction of the sets $P$ and $Q$ costs asymptotically   
$S(R_I/2, V)$, which is the asymptotic size of the Lee sphere of weight $v_i$ and length $m_i$.

In the next step, we check for collisions between $P$ and $Q$.  
The resulting vectors $\be_1 \bA^\top$, respectively $ \bs_1- \be_2\bA^\top $, live in $(\zps)^z$, whereas the second part of the resulting vectors $\be_1 \bC^\top$ and $\bs_3-\be_2 \bC^\top$ lives in $p(\zps)^{(K-k_1)}$,  and from Lemma \ref{unif} we can assume that they are uniformly distributed.  If $v_1,v_2$ are chosen coprime to $p$, then this is certainly true.
Therefore, we have to check asymptotically
$2S(R_I/2,V)-L-R_I+R_1$ many times.
 
Finally, since the success probability of one iteration is given by 
\begin{equation*}
F_L(m_1, v_1,p^s) F_L(m_2, v_2,p^s) F_L(n-K-z, t-v_1-v_2,p^s) F_L(n,t,p^s)^{-1},    
\end{equation*}
the asymptotic number of iterations is given by 

$$ -2S(R_I/2, V)- S(1-R_I-L, T-2V) + S(1,T).$$
\end{proof}

We remark that, if we set $m_1= K$ and thus $m_2=0$, and ask for Lee weight $v_1=v$ and $v_2=0$  in the above algorithm, we get Lee-Brickell's approach, and if we  further choose $\ell=0$ and $v=t$ we get Prange's approach, costing asymptotically $$S(1,T)- S(1-R_I,T).$$

\subsection{$s$-Blocks Algorithm}\label{sblocks}
In this section, we present an algorithm that takes advantage of the structure of the  parity-check matrix over the ring $\zps$. For a code of subtype $(k_1, \ldots, k_s)$ and rank $K$ 
we set $k_{s+1} = n-K$ and bring the parity-check matrix into the systematic form in \eqref{systematicformH}.
The idea, in this algorithm, is to split the error vector into $s+1$ parts, where the first $s$ parts belong to the information set and the last part is outside the information set. Then, we go through all the error vectors having weight $v$ in the information set and $t-v$ outside the information set. In order to go through all such error vectors, we fix a weak compositions of $v$ into $s$ parts, which represent the weight distribution of the first $s$ parts of the error vector. Let us denote by $W_s(v)$ the set of all weak compositions of $v$ into  $s$ parts.

For simplicity, we assume that the information set is $I = \{1,\ldots,K\}$. Let $(v_1,\ldots,v_s)$ be a weak composition of $v$ into $s$ positive integers, and let $v_{s+1} = t-v$.  Therefore, the error vector is of the form $\be=(\be_1,\ldots,\be_{s+1})$, where $\be_{i} \in \left(\vZ / p^s \vZ\right)^{k_i}$ with $\lweight{\be_i} = v_i$ for each $i\in \{1,\ldots,s+1\}$. We first bring the parity-check matrix $\bH$ into the systematic form \eqref{systematicformH}, by multiplying it with an invertible matrix $\bU$. Thus, if we also partition the syndrome $\bs$ into parts of the same size as the (row-)parts of $\bU\bH$, we obtain the following situation:
\begin{align*}
\bU\bH \be^\top &=  \begin{pmatrix} 
\bA_{1,1} & \bA_{1,2} & \cdots & \bA_{1,s-1} & \bA_{1,s} & \Id_{k_{s+1}} \\
p \bA_{2,1} & p \bA_{2,2} & \cdots & p \bA_{2,s-1} & p \Id_{k_s} & \bzero_{k_s \times k_{s+1}} \\
p^2 \bA_{3,1} & p^2 \bA_{3,2} & \cdots & p^2 \Id_{k_{s-1}} & \bzero_{k_{s-1} \times k_s} & \bzero_{k_{s-1} \times k_{s+1}} \\
\vdots & \vdots &  & \vdots & \vdots & \vdots \\
p^{s-1} \bA_{s,1} & p^{s-1} \Id_{k_2} & \ldots & \bzero_{k_{2} \times k_{s-1}} & \bzero_{k_2 \times k_s} & \bzero_{k_2 \times k_{s+1}} 
\end{pmatrix} \begin{pmatrix} \be_1^\top \\ \be_2^\top  \\ \be_3^\top \\ \vdots  \\ \be_{s+1}^\top \end{pmatrix} \\ &= \begin{pmatrix} \bs_1^\top \\ p\bs_2^\top \\ \vdots \\ p^{s-1}\bs_s^\top \end{pmatrix} = \bU \bs^\top.
\end{align*}
From this we obtain $s$ conditions
\begin{align}
& \bA_{1,1} \be_1^\top + \bA_{1,2} \be_2^\top + \cdots + \bA_{1,s}\be_{s}^\top + \be_{s+1}^\top = \bs_1^\top \label{eq:cond1} \\ 
& p  \left( \bA_{2,1}\be_1^\top + \bA_{2,2}\be_2^\top + \cdots + \bA_{2,s-1} \be_{s-1}^\top + \be_{s}^\top \right) = p \bs_2^\top \label{eq:cond2} \\
& \vdots \nonumber \\
& p^{s-2}  \left(\bA_{s-1,1} \be_1^\top + \bA_{s-1,2} \be_2^\top + \be_3^\top \right) = p^{s-2} \bs_{s-1}^\top \label{eq:conds-1} \\
& p^{s-1}  \left( \bA_{s,1}\be_1^\top + \be_2^\top \right) = p^{s-1} \bs_s^\top. \label{eq:conds}
\end{align}
Next, we go through all error vectors of selected weight distribution and check the above conditions using backward recursion.  For this, we use the following notation to expand vectors over $\vZ/ p \vZ$: for a vector $\bx \in \left(\vZ/ p^s \vZ\right)^\ell$, we write $\bx^{(i)} = \bx \pmod{p^i}$ for all $i \in \{1,\ldots,s\}$, and similarly for a matrix $\bB \in \left(\vZ/ p^s \vZ\right)^{m \times \ell}$, we write $\bB^{(i)} = \bB \pmod{p^i}$ for all $i \in \{1,\ldots,s\}$.  Moreover, we define  $ \bx^{(i)} + p^{i} \left( \vZ / p^s \vZ\right)^\ell =\{\bx^{(i)} + p^{i} \by \mid \by \in \left( \vZ / p^s \vZ\right)^\ell\} $ for any $i \in \{1,\ldots,s\}$.

In the next step, we iterate over all the vectors $\be_1 \in \left(\vZ/ p^s \vZ\right)^{k_1}$ having Lee weight $\lweight{\be_{1}}=v_1$. Now, for a given $\be_1$,  Condition \eqref{eq:conds} defines $${\be_2^{(1)}}^\top := {\bs_s^{(1)}}^\top - \bA_{s,1}^{(1)}{\be_1^{(1)}}^\top.$$ Thence, we iterate over all the vectors $\be_2 \in \be_2^{(1)} + p \left(\vZ/ p^s \vZ\right)^{k_2}$. If $\be_2$ has Lee weight $\lweight{\be_2} = v_2$, we proceed; otherwise, we choose another $\be_2$. Now, as a consequence of Condition \eqref{eq:conds-1}, we obtain $$ {\be_3^{(2)}}^\top:={\bs_{s-1}^{(2)}}^\top - \bA_{s-1,1}^{(2)}{\be_1^{(2)}}^\top -\bA_{s-1,2}^{(2)}{\be_2^{(2)}}^\top.$$ Then, we iterate over all the vectors $\be_3 \in \be_3^{(2)} + p^2 \left(\vZ/ p^s \vZ\right)^{k_3}$ and check whether any has the correct Lee weight $v_3$.

We proceed in this fashion until we obtain the $\be_{s+1}$ from Condition \eqref{eq:cond1} in the following way 
$$\be_{s+1}^\top:=\bs_{1}^\top - \sum_{i=1}^s \bA_{1,i}\be_i^\top. $$ And finally we check whether $\be_{s+1}$ has the remaining Lee weight $v_{s+1}$.

The choice of $(v_1, \ldots, v_s) \in W_s(v)$ are inputs of the algorithm that can be optimized, i.e., chosen in such a way that the cost is the lowest.

The details are provided in Algorithm \ref{algo:sblock}.

\begin{algorithm}[h!]
\caption{$s$-Blocks Algorithm over  $\zps$ in the Lee metric}\label{algo:sblock}
\begin{flushleft}
Input: $\bH \in (\zps)^{(n-k_1) \times n}$, $\bs \in(\zps)^{n-k_1}$, $t \in \mathbb{N}$, $v < \min \{ K \lfloor \frac{p^s-1}{2} \rfloor, t\}$ and $(v_1,\ldots,v_s) \in W_s(v)$\\ 
Output: $\be \in(\zps)^n$ with $\be\bH^\top=\bs $ and $\lweight{\be}=t$.
\end{flushleft}
\begin{algorithmic}[1]
\State Choose an information set $I\subset \{1, ...,n\}$  of size $K$.
\State Compute $\bU\in(\zps)^{(n-k_1) \times (n-k_1)}$ and an $n \times n$ permutation matrix $\bP$, such that $\bU\bH \bP$ is in the systematic form as in \eqref{systematicformH}.
 \State Compute $ \bs\bU^\top = \begin{pmatrix}
 \bs_1 & p \bs_2 & \cdots & p^{s-1} \bs_s
 \end{pmatrix}$, where $\bs_i  \in (\Z/p^{s-i+1} \Z)^{k_{s-i+2}}$.
    \For{ $\be_1 \in (\zps)^{k_1}$ with $\lweight{\be_1}=v_1$}
        \State ${\be_2^{(1)}}^\top := {\bs_s^{(1)}}^\top - \bA_{s,1}^{(1)}{\be_1^{(1)}}^\top$.
        \For{$\be_2 \in \be_2^{(1)} + p \left(\vZ/ p^s \vZ\right)^{k_2}$}
            \If{$\lweight{\be_2} = v_2$ }
                \State $ {\be_3^{(2)}}^\top:={\bs_{s-1}^{(2)}}^\top - \bA_{s-1,1}^{(2)}{\be_1^{(2)}}^\top -\bA_{s-1,2}^{(2)}{\be_2^{(2)}}^\top$
                \For {$\be_3 \in \be_3^{(2)} + p^2 \left(\vZ/ p^s \vZ\right)^{k_3}$ }
		            \If{$\lweight{\be_3} =v_3$}
		                \State $\ddots$
		                \State ${\be_s^{(s-1)}}^\top = {\bs_2^{(s-1)}}^\top - \sum_{i=1}^{s-1} \bA_{2,i}^{(s-1)} {\be_i^{(s-1)}}^\top$
		                \For{ $\be_{s} \in \be_s^{(s-1)} + p^{s-1} \left( \vZ / p^s \vZ \right)^{k_s}$}
		                    \If{ $\lweight{\be_s}=v_s$}
		                        \State $\be_{s+1}^\top:=\bs_{1}^\top - \sum_{i=1}^s \bA_{1,i}\be_i^\top. $
		                         \If{ $\lweight{\be_{s+1}} = t-v$}
		                             \State Return $\be = (\be_1, \ldots, \be_{s+1})\bP$
		                        \EndIf
		                    \EndIf
		              \EndFor
		          \EndIf
		      \EndFor
		  \EndIf
	   \EndFor
	\EndFor
\State  Start over with Step 1 and a new selection of $I$.
\end{algorithmic}
\end{algorithm}

For the complexity analysis, we fix the real numbers $V_1,\ldots,V_s$ with $ V_i \leq MR_i, 0 \leq \sum_{i=1}^s V_i \leq T$ and $T-\sum_{i=1}^s V_i \leq (1-R_I)M$.
Then we fix the internal algorithm parameters  $v_i$ such that
$ \lim\limits_{n \to \infty} \frac{v_i}{n}=V_i$, and the code parameters $k_i$ such that $\lim\limits_{n \to \infty} \frac{k_i}{n}=R_i$.

\begin{theorem}
 The asymptotic average time complexity of Algorithm \ref{algo:sblock} is
 \begin{align*}
    & S(1,T) -\sum_{i=2}^s S(R_i, V_i)  - S\left(1-R_I, T-\sum_{i=1}^s V_i\right) \\ &  + \max \left\lbrace  \sum_{i = 2}^{j-1} \frac{-R_i (i-1)}{s} +  S(R_i,V_i) \mid 2 \leq j \leq s+1 \right \rbrace.
\end{align*}
\end{theorem}
\begin{proof}

In one iteration we go through all $\be_1 \in \left(\vZ/p^s\vZ\right)^{k_1}$ of Lee weight $v_1$, which are $F_L(k_1,v_1,p^s)$ many. To compute the cost of the subsequent steps, we divide them in groups of three steps, i.e., at the $j$-th stage we first compute $\be_j^{(j-1)}$ according to the syndrome conditions, and then,  we go through all $\be_j \in \be_j^{(j-1)} + p^{j-1} \zpsk{k_j}$ such that they have Lee weight $v_j$. The cost of computing $\be_j^{(j-1)}$ is given by
\[f(\be_j) := \sum_{\ell=1}^{j-1} k_j k_\ell,\]
and the average number of $\be_j$'s we go through in the set $\be_j^{(j-1)} + p^{j-1} \zpsk{k_j}$ having Lee weight $v_j$ is given by
\[g(\be_j) := \frac{p^{k_j (s-j+1)}}{p^{k_j s}} F_L(k_j,v_j,p^s).\]
Thus, the cost of one iteration is given by 
\[C(v_1,\ldots,v_s) := F_L(k_1,v_1,p^s) \left( \sum_{j=2}^{s+1} f(\be_j) \prod_{\ell=2}^{j-1} g(\be_\ell) \right).\]
This corresponds to  the following asymptotic cost
    \begin{align*}
        & \lim\limits_{n \to \infty}\frac{1}{n}  \log_q\left(  C(v_1, \ldots, v_s) \right) \\
        &= \lim\limits_{n \to \infty}\frac{1}{n}  \log_q\left( F_L(k_1, v_1, p^s) \sum_{j=2}^{s+1} \prod\limits_{\ell=2}^{j-1} p^{-k_\ell(\ell-1)} F_L(k_\ell, v_\ell, p^s) \right)\\
        &= S(R_1,V_1) +  \lim\limits_{n \to \infty}\frac{1}{n}  \log_q\left( \sum_{j=2}^{s+1} p^{\sum\limits_{\ell=2}^{j-1} -k_\ell(\ell-1)}  \prod\limits_{\ell=2}^{j-1} F_L(k_\ell,v_\ell,p^s) \right)\\
        &=  S(R_1,V_1) + \max \left\lbrace  \sum_{\ell = 2}^{j-1} \frac{-R_\ell (\ell-1)}{s} +  S(R_\ell,V_\ell) \mid 2 \leq j \leq s+1 \right \rbrace.
    \end{align*}

The   success probability of one iteration is given by
$$\prod_{i=1}^{s+1} F_L(k_i, v_i, p^s) F_L(n,t,p^s)^{-1}, $$
where, again, we set $k_{s+1}= n-K$ and $v_{s+1}= t-v$.
This corresponds to the asymptotic number of iterations 
$$ S(1,T) -\sum_{i=1}^s S(R_i, V_i)  - S\left(1-R_I, T-\sum_{i=1}^s V_i\right) .$$ 
\end{proof}

\begin{remark}
Another variant of this $s$-blocks algorithm can be obtained by looping over all possible  $(v_1, \ldots, v_s) \in W_s(v)$.
\end{remark}

\subsection{Partial Gaussian Elimination Algorithms}

 In this section we adapt the Partial Gaussian Elimination (PGE)  algorithms from  the Hamming metric to the Lee metric, namely \cite{finiasz, mmt, bjmm}. 
 They all follow the same strategy: by only applying partial Gaussian elimination, depending on a parameter $\ell ,$ we are able to solve the original syndrome decoding instance by solving a smaller one.
 
 More in detail, we find a matrix $\bU \in \zpsk{(n-k_1) \times (n-k_1)}$, such that $$\bU\bH = \begin{pmatrix}
\bA & \Id_{n-K-\ell} \\
\bB & \bzero
\end{pmatrix},$$ 
where $\bA \in \zpsk{(n-K-\ell) \times (K+\ell)}$ and  $\bB \in \zpsk{(K+\ell-k_1) \times (K+\ell)}$. 
Thus, partitioning the error vector $\be$ accordingly, the syndrome decoding problem becomes 
$$ \bU\bH \be^\top =  \begin{pmatrix}
\bA & \Id_{n-K-\ell} \\
\bB & \bzero
\end{pmatrix} \begin{pmatrix}
\be_1^\top \\ \be_2^\top 
\end{pmatrix} = \begin{pmatrix}
\bs_1^\top \\ \bs_2^\top
\end{pmatrix},$$
where $\bs_1 \in \zpsk{n-K-\ell}, \bs_2 \in \zpsk{K+\ell-k_1}$, and $\be_1 \in \zpsk{K+\ell}$ has Lee weight $v$ and $\be_2 \in \zpsk{n-K-\ell}$ has Lee weight $t-v$.
Thus, the initial decoding problem splits into two equations:
\begin{align}
    \bA\be_1^\top + \be_2^\top &= \bs_1^\top, \label{genbrith1cond} \\
    \bB\be_1^\top &= \bs_2^\top. \label{genbirth2cond}
\end{align}
Note that Condition \eqref{genbirth2cond} is again a syndrome decoding instance but with smaller parameters. It is enough to solve this smaller instance and then to compute $\be_2= \bs_1-\be_1\bA^\top$, i.e., such that Condition \eqref{genbrith1cond} is satisfied, and check whether it has the remaining weight $t-v$. The technique to solve the smaller instance will also impact the whole algorithm, as it contributes to the success probability of one iteration. 

In order to solve the smaller syndrome decoding instance, several techniques have been introduced in the Hamming metric, such as Wagner's approach and the representation technique.

\subsubsection{Wagner's Approach}\label{wagner}
The main idea of  Wagner's approach \cite{wagner}, which has been applied to the syndrome decoding problem in \cite{finiasz}, on $a$ levels, is to partition  the searched error vector into $2^a$ subvectors and to store them in a list together with their corresponding partial syndromes. After this, the lists are merged in a special way.    
We will now provide a description of Wagner's approach in more detail. 

The smaller syndrome decoding instance is given by $\bB \in \zpsk{(K+\ell-k_1) \times (K+\ell)}$, $\bs_2 \in \zpsk{K+\ell-k_1}$  and we are searching for an error vector $\be_1 \in \zpsk{K + \ell}$ of Lee weight $v$, such that $\bB \be_1^\top = \bs_2^\top$.

Let $a$ be a positive integer. To ease the notation in Wagner's algorithm  let us assume that $K+\ell$ and $v$ are divisible by $2^a$ (if this is not the case, a small tweak using for example $\left\lfloor \frac{K+\ell}{2^a} \right\rfloor$ can be applied); then, one defines $2^a$ index sets of size $\frac{K+\ell}{2^a}$: $$I_j := \left\{ (j-1)\frac{K+\ell}{2^a}+1, \ldots, j \frac{K+\ell}{2^a} \right\},$$
for all $j \in \{1, \ldots, 2^a\}$.
We then split the input matrix $\bB$ into $2^a$  submatrices $\bB_{j}^{(0)}$, which have columns indexed by $I_j$, i.e.,
$$\bB = \begin{pmatrix}
\bB_1^{(0)} &  \cdots & \bB_{2^a}^{(0)}
\end{pmatrix}$$
and partition the error vector into the corresponding subsets $I_j$ as $$\be_1 = ( \be_1^{(0)}, \ldots, \be_{2^a}^{(0)}).$$
Let us denote by $$I_b^{(i-1)} = I_{2b-1}^{(i-2)} \cup I_{2b}^{(i-2)}$$ for $b \in  \{1,\ldots,2^{a-i+1}\}, i \in \{2,\ldots,a+1\}$ and $I_{b}^{(0)} = I_b$.
Then, we can define $\bB_{2j-1}^{(i-1)} = \bB_{I_{2j-1}^{(i-1)}}.$

We build the initial lists
$$\mathcal{L}_j^{(0)} := \{ \be_{j}^{(0)} \in \zpsk{(K+\ell)/(2^a)} \mid \lweight{\be_{j}^{(0)}} = v/2^a \},$$ 
for $j \in \{1, \ldots, 2^a\}$.
 Let $$0=u_0 \leq u_1 \leq \cdots \leq u_{a-1} \leq u_a = K+\ell-k_1,$$ where $u_i$ is a positive integer indicating the number of entries, on which the merging procedure on level $i$ is based. We will say that $\bx + \by =_u \bz$, if $\bx+\by$ are equal to $\bz$ on the last $u$ positions.
On the $i$-th level one has $2^{a-i+1}$ input lists $\mathcal{L}^{(i-1)}$ and wants to construct $2^{a-i}$ output lists $\mathcal{L}^{(i)}$. These are constructed in the following way 
\begin{align*}
& \mathcal{L}_j^{(i)}  = \mathcal{L}_{2j-1}^{(i-1)} \concat_{\bzero} \mathcal{L}_{2j}^{(i-1)} \coloneqq \\
&  \{(\be_{2j-1}^{(i-1)}, \be_{2j})^{(i-1)} \mid   \be_b^{(i-1)} \in  \mathcal{L}_b^{(i-1)},  \bB_{2j-1}^{(i-1)} (\be_{2j-1}^{(i-1)})^\top =_{u_i} -\bB_{2j}^{(i-1)} (\be_{2j}^{(i-1)})^\top  \},\end{align*}
for $j \in \{1, \ldots, 2^{a-i}-1\},$ whereas for $j=2^{a-i}$ we have
\begin{align*}
&\mathcal{L}_j^{(i)}  = \mathcal{L}_{2j-1}^{(i-1)} \concat_{\bs_2} \mathcal{L}_{2j}^{(i-1)}  \coloneqq \\
&   \{(\be_{2j-1}^{(i-1)}, \be_{2j}^{(i-1)}) \mid   \be_b^{(i-1)} \in  \mathcal{L}_b^{(i-1)}, \bB_{2j-1}^{(i-1)} (\be_{2j-1}^{(i-1)})^\top =_{u_i} \bs_2^\top - \bB_{2j}^{(i-1)} (\be_{2j}^{(i-1)})^\top   \}.\end{align*}

\begin{algorithm}[h!]
\caption{Merge-concatenate}\label{algo:merge-concat}
\begin{flushleft}
Input: The input lists $\mathcal{L}_1, \mathcal{L}_2$, the positive integers $0<u<k$,  $\bB_1, \bB_2 \in \zpsk{k \times n}$ and  $\bs \in \zpsk{k}$. \\ 
Output: $\mathcal{L} = \mathcal{L}_1 \concat_{\bs} \mathcal{L}_2$.
\end{flushleft}
\begin{algorithmic}[1]
\State Lexicographically sort $\mathcal L_1$  according to the last $u$ positions of $\bB_1 \be_1^\top$ for $\be_1 \in \mathcal L_1$. We also store the last $u$ positions of $\bB_1\be_1^\top$ in the sorted list.
\For{$\be_2 \in \mathcal{L}_2$}
\For{$\be_1 \in \mathcal{L}_1$ with $\bB_1\be_1^\top = \bs^\top - \bB_2 \be_2^\top$ on the last $u$ positions}
    \State $\mathcal L = \mathcal L \cup \{(\be_1,\be_2)\}$.
\EndFor
\EndFor
\State Return $\mathcal{L}.$
\end{algorithmic}
\end{algorithm}

\begin{lemma}   \label{lemma:cost-mc}
The asymptotic of the average cost of Algorithm \ref{algo:merge-concat} is
$$\lim\limits_{n \to \infty} \frac{1}{n} \max \left\{\log_{p^s}  \left(\mid \mathcal{L}_1 \mid\right), \log_{p^s}\left( \mid \mathcal{L}_2\mid \right), \log_{p^s}(\mid \mathcal{L}_1 \mid) + \log_{p^s}(\mid \mathcal{L}_2 \mid ) -u \right\}.$$
\end{lemma}
\begin{proof}
 As a first step, we sort the list $\mathcal{L}_1$ according to the last $u$ positions of $\bB_1 \be_1^\top$. Then, for each $\be_2 \in \mathcal{L}_2$, we compute the last $u$ positions of  $\bs^\top-\bB_2\be_2^\top$ and check for a collision in the sorted list $\mathcal{L}_1$.
The asymptotic cost to compute  $\bB_1\be_1^\top$  on the last $u$ positions for each $\be_1 \in \mathcal{L}_1$, and to sort $\mathcal{L}_1$ is the same: \[\lim\limits_{n\to \infty}\frac{1}{n}\log_{p^s}\left(\mid \mathcal{L}_1 \mid\right).\]
 
Similarly, to  compute $\bs^\top-\bB_2\be_2^\top$  on the last $u$ positions for each $\be_2 \in \mathcal{L}_2$, asymptotically  costs \[\lim\limits_{n\to \infty}\frac{1}{n}\log_{p^s}\left(\mid \mathcal{L}_2 \mid\right).\] 

Finally, the expected number of collision in Step 3 is $|\mathcal{L}_1||\mathcal{L}_2|/p^{su}$.
\end{proof}

After the merge on level $a$, one is left with a final list, which contains now  error vectors $\be_1^{(a)} = (\be_{1}^{(0)}, \ldots, \be_{2^a}^{(0)})$ of Lee weight $v$ such that
$$\bB_1^{(0)} (\be_{1}^{(0)})^\top + \cdots + \bB_{2^a}^{(0)}(\be_{2^a}^{(0)})^\top = \bs_2^\top $$
and can hence solve the smaller syndrome decoding instance.

This algorithm based on Wagner's approach succeeds if the target vector $\be_1$ can be split into subvectors $\be_1^{(0)},\ldots,\be_{2^a}^{(0)}$ of length $(K+l)/2^a$ and weight $v/2^a$, and at each level $i$ the syndrome equations are satisfied on some fixed $u_i$ positions. 
In particular, the success probability is then given by the probability that at least one representation of the target vector lives in the lists $\mathcal{L}^{(i)}$. Let us assume that $v \neq 0.$ 

For the base lists, we have the success probability that we can partition the target vector $\be_1$ into vectors of length $(K+\ell)/2^a$ of Lee weight $v/2^a$ is given by 
\begin{equation}
\frac{F_L(n-K-\ell,t-v,p^s)F_L(\frac{K+\ell}{2^a},\frac{v}{2^a},p^s)^{2^a}}{F_L(n,t,p^s)}.    \label{eq:Wagner_SP}
\end{equation}
 For the lists after one merge, i.e., $\mathcal{L}_i^{(1)}$, we now want that the vectors $\be_{i}^{(1)} \in \mathcal{L}_i^{(1)}$ are such that $\bB_i^{(1)}(\be_i^{(1)})^\top =_{u_1} \mathbf{0}^\top$.
Since $\bB_i^{(1)}, \be_i^{(1)}$ are assumed to be uniformly distributed, we can use Lemma \ref{unif}, to get that 
\begin{enumerate}
    \item If $p \nmid v/2^{a-1}$, then $\be_i^{(1)} \not\in p\zpsk{(K+\ell)/2^{a-1}}$ with probability 1.
    \item If $p \mid v/2^{a-1}$, then $\be_i^{(1)} \not\in p\zpsk{(K+\ell)/2^{a-1}}$ with probability
    $$1- \frac{F_L((K+\ell)/2^{a-1}, v/(p2^{a-1}),p^{s-1})}{F_L((K+\ell)/2^{a-1}, v/2^{a-1},p^s)}.$$
\end{enumerate}
Thus, $\be_i^{(1)} \not\in p\zpsk{(K+\ell)/2^{a-1}}$ with high probability and again due to Lemma \ref{unif}, we get that $\bB_i^{(1)}(\be_i^{(1)})^\top$ is uniformly distributed. Thus, $\bB_i^{(1)}(\be_i^{(1)})^\top =_{u_1} \mathbf{0}^\top$ has probability $p^{-su_1}.$
Note that the condition  $\bB_i^{(1)}(\be_i^{(1)})^\top =_{u_1} \mathbf{0}^\top$ 
is imposed on all $2^{a-1}$ vectors (and for the last one a similar condition holds but summing to $\bs_2$). However, ensuring only half of them satisfy the condition is sufficient, 
because with the condition on the next merges we can ensure that the other half satisfies this, as well.   
For example, since $\bB^{(2)}_j (\be^{(2)}_j)^\top = \bB^{(1)}_{2j-1} (\be^{(1)}_{2j-1})^\top + \bB^{(1)}_{2j} (\be^{(1)}_{2j})^\top =_{u_2} \mathbf{0}^\top$ and $\bB_{2j-1}^{(1)}(\be_{2j-1}^{(1)})^\top =_{u_1} \mathbf{0}^\top$, it directly follows that $\bB_{2j}^{(1)}(\be_{2j}^{(1)})^\top =_{u_1} \mathbf{0}^\top$.

For the next merge, we impose that  $\bB_i^{(2)}(\be_i^{(2)})^\top =_{u_2} \mathbf{0}^\top$, knowing that $$\bB_i^{(2)}(\be_i^{(2)})^\top =_{u_1}  \mathbf{0}^\top$$ from the condition on the lists $\mathcal{L}_i^{(1)},$ thus this results in a probability that the above condition is verified of 
$p^{-s(u_2-u_1)}$  and is imposed on $2^{a-3}$ vectors.

We can continue in the same fashion for the subsequent merges, resulting for the algorithm on $a$ levels in a total probability of 
$$ \prod\limits_{i=1}^{a-1} \left( p^{-s(u_i-u_{i-1})}\right)^{2^{a-i-1}}.$$

We expect
by Lemma \ref{unif}
all lists $\mathcal{L}_b^{(i)}$ to be of size $$  L_i = \frac{L_{i-1}^2}{\left(p^s\right)^{u_i-u_{i-1}}}. $$ Hence, for $i \in \{1, \ldots, a\}$ we have that  $$L_i = \frac{L_0^{2^{i}}}{\left(p^s\right)^{\beta(i)}}, $$
where $\beta_i = u_i + \sum_{j=1}^{i-1} 2^{i-j-1} u_j.$

Observe that the final merge on $K+\ell-k_1$ positions, that is 
$$\bB_1^{(a-1)}(\be_1^{(a-1)})^\top = \bs_2^\top -\bB_2^{(a-1)}(\be_2^{(a-1)})^\top$$ is not considered in the computation of the probability, as here we are not losing  any representation of the target vector, since we are not assuming that some parts sum to $0$.

\begin{algorithm}[h!]
\caption{Wagner on $a$ levels}\label{algo:wagner}
\begin{flushleft}
Input: $0 = u_0 \leq u_1 \leq \cdots \leq u_a= K+ \ell-k_1, 0\leq v \leq t, \bH \in \zpsk{(n-k_1) \times n}$ and $\bs \in \zpsk{n-k_1}$. \\ 
Output: $\be \in \zpsk{n}$ with $\text{wt}_L(\be)= t$ and $\bH\be^\top = \bs^\top.$ 
\end{flushleft}
\begin{algorithmic}[1]
\State Choose an $n \times n$ permutation matrix $\bP$.
\State Find $\bU  \in \zpsk{(n-k_1) \times (n-k_1)}$, such that 
$$\bU\bH\bP = \begin{pmatrix}
\bA & \Id_{n-K-\ell} \\ \bB & \bzero
\end{pmatrix}, $$ where $\bA \in \zpsk{(n-K-\ell) \times (K+\ell)}$ and $\bB \in \zpsk{(K+\ell-k_1) \times (K+\ell)}.$
\State Compute $\bU\bs^\top = \begin{pmatrix}
\bs_1^\top \\ \bs_2^\top
\end{pmatrix},$ where $\bs_1 \in \zpsk{n-K-\ell}, \bs_2 \in \zpsk{K+\ell-k_1}.$
\State Partition $I=\{1, \ldots, K+\ell\}$ into $2^a$ subsets of size $\frac{K+\ell}{2^a}$: $$I_j=\left\{(j-1) \frac{K+\ell}{2^a}+1, \ldots, j\frac{K+\ell}{2^a}\right\}, $$
for $j \in \{1, \ldots, 2^a\}.$
\State Partition $\bB$ into $2^a$ submatrices $\bB_{j},$ containing the columns of $\bB$ indexed by $I_j$, for all $j \in \{1, \ldots, 2^a\}.$
\State Set $$\mathcal{L}_j^{(0)} = \{ \be_{j}^{(0)} \in \zpsk{(K+ \ell)/(2^a)}, \lweight{\be_{j}^{(0)}} = v/(2^a) \}$$ for all $j \in \{1, \ldots, 2^a\}.$
 \For{ $ i \in \{1, \ldots a\}$}
    \For{$j \in \{1, \ldots, 2^{a-i}\}$} 
        \State Compute $\mathcal{L}_j^{(i)} = \begin{cases} \mathcal{L}_{2j-1}^{(i-1)} \concat_{\bzero} \mathcal{L}_{2j}^{(i-1)} & \mbox{for } 1 \leq j < 2^{a-i}, \\
\mathcal{L}_{2j-1}^{(i-1)} \concat_{\bs_2} \mathcal{L}_{2j}^{(i-1)} & \mbox{for } j = 2^{a-i}.
\end{cases}$
    \EndFor
 \EndFor
 \For{$\be_1 \in \mathcal{L}_1^{(a)}$}
 \If{$\lweight{\bs_1-\be_1\bA^\top}=t-v $}
 \State Return $\be = \bP(\be_1, \bs_1-\be_1\bA^\top)$
 \EndIf
 \EndFor
 \State Else start over at step 1.
 \end{algorithmic}
\end{algorithm}

In the following, we compute the asymptotic average cost of Wagner's algorithm for level one and level two. 
This is enough, as it turns out that level one gives the optimal complexity, and increasing the levels does not improve the cost, 
as one can see in the comparison in Section \ref{sec:asymptotic}.
\begin{theorem}
\begin{enumerate}
    \item For Wagner's approach on one level, let us fix the real numbers $V,L$ with $0 \leq V \leq T$ and $0 \leq L \leq 1-R_I$ such that $V \leq  M(R_I+L)$ and $T-V \leq M(1-R_I-L)$. 
We fix the internal algorithm parameters $v, \ell$, such that
$ \lim\limits_{n \to \infty} \frac{v}{n}=V, \lim\limits_{n \to \infty} \frac{\ell}{n}=L$.
The asymptotic average time complexity of Wagner's approach on one level is then given by 
\begin{gather*}   S(1,T)-S(1-R_I-L,T-V)-2S((R_I+L)/2, V/2) + \\
  \max\{S((R_I+L)/2,V/2), 2S((R_I+L)/2, V/2)-(R_I+L-R_1)\}.
\end{gather*}
\item For Wagner's approach on two levels, let us fix  additionally the real number $U$, such that $0\leq U \leq R_I+L-R_1$ and we fix the internal algorithm parameters $u_1$, such that $\lim\limits_{n \to \infty} \frac{u_1}{n}=U.$
The asymptotic average time complexity of Wagner's approach on two levels is then given by
\begin{gather*}    S(1,T)-S(1-R_I-L,T-V)-4S((R_I+L)/4, V/4) + U + \\
\max\{S((R_I+L)/4,V/4), 2S((R_I+L)/4, V/4)-U,  \\ 4 S((R_I+L)/4, V/4)-(R_I+L-R_1)-U\}.
\end{gather*}  
\end{enumerate}
\label{thm:wagner}
\end{theorem}

\begin{proof} 
 \begin{enumerate}
     \item The initial lists $\mathcal{L}_1^{(0)}, \mathcal{L}_2^{(0)}$ are  of size $L_0 = F_L\left( \frac{K+\ell}{2}, \frac{v}{2}, p^s\right)$.
     We now have to compute $\mathcal{L}_1^{(1)}= \mathcal{L}_1^{(0)} \concat_{\bs_2} \mathcal{L}_2^{(0)}. $
     Which due to Algorithm \ref{algo:merge-concat} costs $$ \max\{S((R_I+L)/2,V/2),2S((R_I+L)/2, V/2)-(R_I+L-R_1) \}.$$ 
     Finally, we check for all $\be_1 \in \mathcal{L}_1^{(1)}$ if  $\lweight{\bs_1-\be_1\bA^\top}=t-v$, which asymptotically costs 
     $$2S((R_I+L)/2, V/2)-(R_I+L-R_1),$$ which is the asymptotic size of $ \mathcal{L}_1^{(1)}.$
     The success probability is given by $$\frac{F_L(n-K-\ell,t-v,p^s) F_L((K+L)/2,v/2,p^s)^2}{F_L(n,t,p^s)},$$ and hence the asymptotic number of iterations is $$S(1,T)-S(1-R_I-L,T-V)-2S((R_I+L)/2, V/2).$$
     \item In the level 2 case, we  have 4 base lists, $\mathcal{L}_b^{(0)}$ of size $L_0$, for $b \in \{1, \ldots, 4\}.$
     We first merge $\mathcal{L}_1^{(1)}= \mathcal{L}_1^{(0)} \concat_{\mathbf{0}} \mathcal{L}_2^{(0)}$, and  $\mathcal{L}_2^{(1)}= \mathcal{L}_3^{(0)} \concat_{\mathbf{s}_2} \mathcal{L}_4^{(0)}$,
which both cost $$\max\{S((R_I+L)/4,V/4), 2S((R_I+L)/4,V/4) - U \}.$$
The final merge $\mathcal{L}_1^{(2)} = \mathcal{L}_1^{(1)} \concat_{\mathbf{s}_2} \mathcal{L}_2^{(2)} $ costs $2S((R_I+L)/4,V/4)$ which is the asymptotic size of $\mathcal{L}_b^{(1)}$ for $b \in \{1,2\}.$
To check for all $\be_1 \in \mathcal{L}_1^{(2)}$ if  $\lweight{\bs_1-\be_1\bA^\top}=t-v$, asymptotically costs 
     $$ 4 S((R_I+L)/4, V/4)-(R_I+L-R_1)-U,$$ which is the asymptotic size of $ \mathcal{L}_1^{(2)}.$
     
     The success probability is composed of the probability that $\be$ has the assumed weight distribution and that a merging on $u_1$ positions results in the sought-after $\be$, which is given by $$\frac{F_L(n-K-\ell,t-v,p^s) F_L((K+L)/4,v/4,p^s)^4}{F_L(n,t,p^s)p^{u_1}},$$ and hence the asymptotic number of iterations is $$S(1,T)-S(1-R_I-L,T-V)-4S((R_I+L)/4, V/4)+U.$$

 \end{enumerate}

\end{proof}

\subsubsection{Representation Technique}\label{rep}

In contrast to the above description of Wagner's approach, in \cite{bjmm} the authors allow the subvectors of the error vector to overlap; this is called the subset sum representation technique. Note that the approach in \cite{mmt} is similar to \cite{bjmm}, but in the former the vectors have no overlap.  In this section we will adapt the representation technique approach  of \cite{bjmm} to solve the smaller syndrome decoding instance to the Lee metric.

In this algorithm, we create the list $\mathcal L$ of vectors having Lee weight $v$ by merging two lists $\mathcal L_1, \mathcal L_2$ containing vectors of Lee weight $v/2 + \varepsilon$, where $\varepsilon$ represents the overlapping part that cancels out after adding the vectors. Thus, if we want to reach a $\bx \in \zpsk{n}$ of Lee weight $v$, we can write it as $\bx=\by+\bz$,  for $\by$ and $\bz$ having Lee weight $v/2+ \varepsilon$. For $M = \left\lfloor \frac{p^s}{2}\right\rfloor$, this $\varepsilon$ can be between $0$ and $Mn-v$, as we can reach weight $v$ be using two times $v/2$ and the room we are left with for overlappings has weight $Mn-v$. Algorithm \ref{algo:merge} describes the process of merging with overlapping. We apply this merging process at each level, similar to Wagner's approach.

The weight in the overlapping parts is described by the choice of the positive integers $\varepsilon_0,\ldots,\varepsilon_{a-1}$. Let $v_0 = \frac{v}{2^a} + \sum_{b=0}^{a-1} \frac{\varepsilon_b}{2^b}$. We start by computing the base lists
\[\mathcal{L}_j^{(0)} = \left \{  \be_j^{(0)} \in \zpsk{K + \ell } \mid \lweight{\be_j^{(0)}} = v_0 \right\},\]
for $j \in \{1,\ldots,2^a\}$.
The algorithm is divided again into $a$ levels, and at level $i \in \{1,\ldots,a\}$ the vectors have Lee weight $\varepsilon_{i-1}$ in the overlapping part, which cancels out after adding two vectors, and the rest of the merging works as in Wagner's approach. In other words, we merge the lists $\mathcal{L}_{2j-1}^{(i-1)}$ and $\mathcal{L}_{2j}^{(i-1)}$ to obtain a list $\mathcal{L}_{j}^{(i)}$ such that the overlapping part cancels out, where the lists are given by
\begin{align*}
    \mathcal{L}_{j}^{(i)} = \left \{ \be_{j}^{(i)}  \in \zpsk{K + \ell } \mid \lweight{\be_j^{(i)}} = v_i,  \bB {\be_j^{(i)}}^\top =_{u_i} \bs_2^\top   \right \}.
\end{align*}
Here $v_i= \frac{v}{2^{a-i}} + \sum_{b=i}^{a-1} \frac{\varepsilon_{b}}{2^{b-i}}$,
for each $j \in \{1,\ldots,2^{a-i}\}$, and $$1 \leq u_1 \leq \cdots \leq  u_a = K+\ell-k_1$$ is an non-decreasing sequence of integers. 

The resulting list $\mathcal{L}_j^{(i)}$ is expected to have size 
$$L_i = F_L(K+\ell, v_i,p^s)/(p^s)^{u_i}. $$ 

In order to ensure that at least one representation $(\be_1^{(i)}, \be_2^{(i)})$ of a vector $\be_1^{(i+1)} = \be_1^{(i)}+ \be_2^{(i)}$ is in $\mathcal{L}_1^{(i)} \times \mathcal{L}_2^{(i)}$, we have to choose the $u_i$'s in a more restrictive way. For example any value smaller than the logarithm of the number of representations will suffice.

For a vector $\bx$ of weight $v$ of length $K+\ell$ we observe that the number of representations  $\by +\bz$ of weight $v/2 + \varepsilon$ depends on the choice of $\bx$ and thus does not have a generic formula. Thus, we can either work with a lower bound or with an expected number of representations. Using a lower bound ensures that for any error vector the algorithm works; however, on average, the expected number of representations should be enough to ensure this as well.

Let us denote by $C(w,\sigma,\eta)$ the number of compositions of $w$ into $\sigma$ parts such that each part is at most $\eta$, which can be computed as 
\[
C(w,\sigma,\eta)=\sum_{j=0}^{\min\{\sigma,\lfloor\frac{w-\sigma}{\eta}\rfloor\}}{ (-1)^j\binom{\sigma}{j} \binom{w-j\eta-1}{\sigma-1}},
\]
as shown in \cite{Abramson1976}.

\begin{lemma}\label{nrrep}
Let $\bx \in \zpsk{K+\ell}$ with $\lweight{\bx}=v$. The number of representations $\bx=\by+\bz$ with $\lweight{\by}=\lweight{\bz}=v/2+\epsilon$ is greater than or equal to 
$$2 \sum\limits_{\sigma=1}^{\min\{\varepsilon, K+\ell\}} \binom{K+\ell}{\sigma}C(\varepsilon, \sigma, \lfloor p^s/4 \rfloor).$$
\end{lemma}
\begin{proof}
Let us denote by $M= \lfloor \frac{p^s}{2} \rfloor.$ To get a lower bound, one construction of a representation is enough.
 One can write each representation $\by+\bz$ as 
$$ (\by_1+\by_2)+ (\bz_1+ \bz_2),$$
where $\lweight{\by_1} = \lweight{\bz_1}=v/2$ and $\lweight{\by_2}=\lweight{\bz_2}=\varepsilon.$
We will focus on a particular representation, where we half the weight of each entry of $\bx$. Thus, it is enough to consider one entry of $\bx$, e.g. the first $\bx_1$. Clearly, the choice of $\by_1$ then completely defines the choice of $\bz_1$  and also the choice of $\by_2$ determines the choice of $\bz_2= -\by_2$. For each non-zero entry (and thus also for the first entry) of $\by_1$ there exist always two choices, namely for the first entry these are  $
    \lfloor \bx_1/2 \rfloor$ and $
    \lceil (\bx_1-p^s)/2 \rceil.
$
The choice of $\by_2$ is then restricted, since we want to add weight by computing $\by_1+\by_2$, but for certain choices of $\by_2$ this reduces the weight. In fact, for any choice of first entry of $\by_1$  (assuming it is non-zero), we have at least $\lfloor M/2 \rfloor$ choices for the first entry of $\by_2$. To see this, imagine the first entry of $\by_1$ called $\alpha$ to be in $[-M,M].$ We denote by $\beta$ the first entry of $\by_2$. If $\alpha>0$, then in order for $\beta+\alpha \in [0, M]$ we must choose $\beta \in [0, \lfloor M/2 \rfloor]$, and if $\alpha <0$, then in order for $\beta + \alpha \in [-M,0]$ we must choose $\beta \in [-\lfloor M/2 \rfloor, 0]$. Thus, with this particular construction of $\by$ and $\bz$ we are allowed to compose the weight $\varepsilon$ into $\sigma$ parts, denoting the support size, of size up to $\lfloor M/2 \rfloor$.
 
\end{proof}

Thus, using this lower bound we choose $$ u_i = \left\lfloor \log_{p^s} \left(2 \sum\limits_{\sigma=1}^{\min\{\varepsilon, K+\ell\}} \binom{K+\ell}{\sigma}C(\varepsilon, \sigma, \lfloor p^s/4 \rfloor)\right) \right\rfloor.$$
As mentioned earlier, another approach for the choice of $u_i$ can be based on the expected number of representations. However, due to the difficulty of calculating the expected number of representations, we leave this approach for future works.

\begin{figure}[ht]
    \centering
    
    \begin{tikzpicture}[emptybox/.style={shape=rectangle, rounded corners,
    draw, align=center, minimum height = 15pt, minimum width = 10.8cm, anchor =west}, 
    shaded/.style = {shape=rectangle, minimum height = 14pt, fill = black!20, minimum width = 1cm, anchor = west},
    shadede/.style = {shape=rectangle, minimum height = 14pt, fill = black!20, minimum width = 0.5cm, anchor = west},
    shadede1/.style = {shape=rectangle, minimum height = 14pt, fill = black!20, minimum width = 1.1cm, anchor = west}]

\node[left] at (-5,0) {$\be_1^{(0)} = $};
\node[emptybox] at (-5,0) {};
\node[shaded] at (-4.9,0) {$\frac{v}{4}$};
\node[shadede, pattern= north east hatch, pattern color=blue!15, hatch distance=8pt, hatch thickness=2pt] at (-3.8,0) {$\varepsilon_0$};
\node[shadede] at (-1.05,0) {$\frac{\varepsilon_1}{2}$};

\node[left] at (-5,-0.6) {$\be_2^{(0)} = $};
\node[emptybox] at (-5,-0.6) {};
\node[shaded] at (-2.7,-0.6) {$\frac{v}{4}$};
\node[shadede, pattern= north east hatch, pattern color=blue!15, hatch distance=8pt, hatch thickness=2pt] at (-3.8,-0.6) {$\varepsilon_0$};
\node[shadede] at (-1.6,-0.6) {$\frac{\varepsilon_1}{2}$};

\node[left] at (-5,-1.2) {$\be_3^{(0)} = $};
\node[emptybox] at (-5,-1.2) {};
\node[shaded] at (0.1,-1.2) {$\frac{v}{4}$};
\node[shadede, pattern= north east hatch, pattern color=blue!15, hatch distance=8pt, hatch thickness=2pt] at (1.2,-1.2) {$\varepsilon_0$};
\node[shadede] at (-1.05,-1.2) {$\frac{\varepsilon_1}{2}$};

\node[left] at (-5,-1.8) {$\be_4^{(0)} = $};
\node[emptybox] at (-5,-1.8) {};
\node[shaded] at (2.3,-1.8) {$\frac{v}{4}$};
\node[shadede, pattern= north east hatch, pattern color=blue!15, hatch distance=8pt, hatch thickness=2pt] at (1.2,-1.8) {$\varepsilon_0$};
\node[shadede] at (-1.6,-1.8) {$\frac{\varepsilon_1}{2}$};

\node[left] at (-5,-3.6) {$\be_1^{(1)} = $};
\node[right] at (-5,-3.6) {$\be_1^{(0)} + \be_2^{(0)}$};
\node[left] at (-5,-4.2) {=};
\node[emptybox] at (-5,-4.2) {};
\node[shaded] at (-4.9,-4.2) {$\frac{v}{4}$};
\node[shaded] at (-2.7,-4.2) {$\frac{v}{4}$};
\node[shadede1, pattern= north east hatch, pattern color=blue!15, hatch distance=8pt, hatch thickness=2pt] at (-1.6,-4.2) {$\varepsilon_1$};

\node[left] at (-5,-4.8) {$\be_2^{(1)} = $};
\node[right] at (-5,-4.8) {$\be_3^{(0)} + \be_4^{(0)}$};
\node[left] at (-5,-5.4) {=};
\node[emptybox] at (-5,-5.4) {};
\node[shaded] at (0.1,-5.4) {$\frac{v}{4}$};
\node[shaded] at (2.3,-5.4) {$\frac{v}{4}$};
\node[shadede1, pattern= north east hatch, pattern color=blue!15, hatch distance=8pt, hatch thickness=2pt] at (-1.6,-5.4) {$\varepsilon_1$};

\node[left] at (-5,-7.2) {$\be_1^{(2)} = $};
\node[right] at (-5,-7.2) {$\be_1^{(1)} + \be_2^{(1)}$};
\node[left] at (-5,-7.8) {=};
\node[emptybox] at (-5,-7.8) {};
\node[shaded] at (-4.9,-7.8) {$\frac{v}{4}$};
\node[shaded] at (-2.7,-7.8) {$\frac{v}{4}$};
\node[shaded] at (0.1,-7.8) {$\frac{v}{4}$};
\node[shaded] at (2.3,-7.8) {$\frac{v}{4}$};

\draw[dashed, rounded corners] (-6.2,-6.2) rectangle (6.2,-8.4);
\node[below] at (0,-6.3) {Level 2};
\draw[dashed, rounded corners] (-6.2,-2.6) rectangle (6.2,-6);
\node[below] at (0,-2.7) {Level 1};
\draw[dashed, rounded corners] (-6.2,1) rectangle (6.2,-2.4);
\node[below] at (0,0.9) {Level 0};

\end{tikzpicture}

    \caption{Illustration of the error vectors at each level of the representation technique algorithm for two levels. At each level, the striped region represents the overlapping part.}
    \label{fig:rep_tech}
\end{figure}

\begin{algorithm}[h!]
\caption{Merge}\label{algo:merge}
\begin{flushleft}
Input: The input lists $\mathcal{L}_1, \mathcal{L}_2$, the positive integers $0<u<k$ and  $0 \leq v \leq \left\lfloor \frac{p^s}{2} \right\rfloor n$, the matrix  $\bB \in \zpsk{k \times n}$ and the syndrome $\bs \in \zpsk{k}$. \\ 
Output: $\mathcal{L} = \mathcal{L}_1 \bowtie \mathcal{L}_2$.
\end{flushleft}
\begin{algorithmic}[1]
\State Lexicographically sort $\mathcal L_1$  according to the last $u$ positions of $\bB \be_1^\top$ for $\be_1 \in \mathcal L_1$. We also store the last $u$ positions of $\bB\be_1^\top$ in the sorted list. 
\For{$\be_2 \in \mathcal{L}_2$}
\For{$\be_1 \in \mathcal{L}_1$ with $\bB \be_1^\top = \bs^\top - \bB \be_2^\top$ on the last $u$ positions}
    \If{$\lweight{\be_1+\be_2} = v$}
        \State $\mathcal L = \mathcal L \cup \{\be_1 + \be_2\}$.
    \EndIf
\EndFor
\EndFor
\State Return $\mathcal{L}.$
\end{algorithmic}
\end{algorithm}

\begin{lemma} \label{lemma:merge-cost}
The asymptotic average time complexity of the merge algorithm (Algorithm \ref{algo:merge}) is given by
$$ \lim_{n \to \infty} \frac{1}{n} \max\{ \log_{p^s}(L_1), \log_{p^s}(L_2),  \log_{p^s}(L_1) + \log_{p^s}(L_2) -u \} , $$ where $L_i = |\mathcal L_i |$ for $i = 1,2$. 
\end{lemma}
\begin{proof}
We sort the list $\mathcal{L}_1$  while computing the vectors $\bB \be_1^\top$ for each $\be_1 \in \mathcal{L}_1$.  Once we know the $u$ positions of this vector, we append $\be_1^\top$ to the sorted list containing previously computed vectors. Let $L_1$ be the size of $\mathcal{L}_1$. Then, the average cost of sorting $\mathcal{L}_1$ according to the last $u$ positions of $\bB \be_1^\top$ is 
asymptotically $$\lim_{n \to \infty} \frac{1}{n} \left( \log_{p^s}(L_1) \right). $$

Then for each $\be_2 \in \mathcal L_2$, we compute the last $u$ positions of $\bs^\top - \bB \be_2^\top$ and find $\be_1$ in the sorted list $\mathcal L_1$ satisfying $\bB \be_1^\top =_u \bs^\top - \bB \be_2^\top$.
This step will cost on average 
$$\lim_{n \to \infty} \frac{1}{n} \left( \log_{p^s}(L_2) \right). $$ 
 
Next, for each collision we compute the vector $\be_1 + \be_2$, which is asymptotically the most dominant step. Using the estimation for the expected number of collisions, we get that the overall average cost of the algorithm is asymptotically given by 
 
$$ \lim\limits_{n \to \infty} \frac{1}{n} \left(\log_{p^s}(L_1) + \log_{p^s}(L_2) -u\right). $$ 
\end{proof}

Algorithm \ref{algo:merge} is applied for all levels $i<a$ whereas, on the last level $a$, we do a similar merging algorithm, with a small twist: instead of Step 5 of Algorithm \ref{algo:merge}, we check whether 
$$\text{wt}_L(\bs_1-(\be_1 + \be_2)\bA^\top)=t-v,$$
which is necessary to guarantee that the considered solution $(\be_1+ \be_2)$ of the smaller SDP instance can be extended to a solution of the whole SDP instance,  being $(\be_1+\be_2, \bs_1-(\be_1 + \be_2)\bA^\top)$. Taking this condition into account inside the merge algorithm prevents us from storing the final list, thus decreasing the cost of the algorithm.
Thus, in the last level, as soon as this condition is satisfied, one can abort and return the found error vector; we denote this merging process as $$ \mathcal{L}_1 \bowtie_{\bs_2,\bs_1,\bA,w}^{(a)} \mathcal{L}_2.$$

\begin{algorithm}[h!]
\caption{Last Merge}\label{algo:last-merge}
\begin{flushleft}
Input: The input lists $\mathcal{L}_1, \mathcal{L}_2$, the positive integers $w,v,0<u<k$,  $\bB \in \zpsk{k \times n}, \bs_2 \in \zpsk{k}$ and $\bs_1 \in \zpsk{r}, \bA \in \zpsk{r \times 2n}$. \\ 
Output: $\be \in \mathcal{L}_1 \bowtie_{\bs_2,\bs_1,\bA,w}^{(a)} \mathcal{L}_2$.
\end{flushleft}
\begin{algorithmic}[1]
\State Lexicographically sort $\mathcal L_1$  according to the last $u$ positions of $\bB \be_1^\top$ for $\be_1 \in \mathcal L_1$. We also store the last $u$ positions of $\bB\be_1^\top$ in the sorted list. 
\For{$\be_2 \in \mathcal{L}_2$}
\For{$\be_1 \in \mathcal{L}_1$ with $\bB \be_1^\top = \bs_2^\top - \bB \be_2^\top$ on the last $u$ positions}
    \If{$\lweight{\be_1+\be_2} = v$ and $\lweight{\bs_1-(\be_1+ \be_2)\bA^\top}=w$}
    \State Return $(\be_1,\be_2, \bs_1-(\be_1+ \be_2)\bA^\top)$.
    \EndIf
\EndFor
\EndFor
\end{algorithmic}
\end{algorithm}

Note that the cost of this last merge is the same as the cost of Algorithm \ref{algo:merge}.
\begin{corollary} \label{last-merge-cost}
The asymptotic average time complexity of the last  merge (Algorithm \ref{algo:last-merge}) is given by
$$ \lim_{n \to \infty} \frac{1}{n} \max\{ \log_{p^s}(L_1), \log_{p^s}(L_2),  \log_{p^s}(L_1) + \log_{p^s}(L_2) -u \} , $$ and $L_i = |\mathcal L_i |$ for $i = 1,2$. 
\end{corollary}

\begin{algorithm}[h!]
\caption{Representation Technique on $a$ levels}\label{algo:bjmm}
\begin{flushleft}
Input: $0\leq \ell\leq n-K$, $0 \leq \varepsilon_i \leq M(K+\ell)-v_{i+1}$, for all $i \in \{0, \ldots, a-1\}$  where $v_i = v/(2^{a-i}) + \sum_{b=i}^{a-1} \frac{\varepsilon_{b}}{2^{b-i}}$,  $0 = u_0 \leq u_1 \leq \cdots \leq u_a= K+\ell-k_1, t,v, \bH \in \zpsk{(n-k_1) \times n}$ and $\bs \in \zpsk{n-k_1}$. \\ 
Output: $\be \in \zpsk{n}$ with $\text{wt}_L(\be)= t$ and $\bH\be^\top = \bs^\top.$ 
\end{flushleft}
\begin{algorithmic}[1]
\State Choose an $n \times n$ permutation matrix $\bP$.
\State Find $\bU  \in \zpsk{(n-k_1) \times (n-k_1)}$, such that 
$$\bU\bH\bP = \begin{pmatrix}
\bA & \Id_{n-K-\ell} \\ \bB & \bzero
\end{pmatrix}, $$ where $\bA \in \zpsk{(n-K-\ell) \times (K+\ell)}$ and $\bB \in \zpsk{(K+\ell-k_1) \times (K+\ell)}.$
\State Compute $\bU\bs^\top = \begin{pmatrix}
\bs_1^\top \\ \bs_2^\top
\end{pmatrix},$ where $\bs_1 \in \zpsk{n-K-\ell}, \bs_2 \in \zpsk{K+\ell-k_1}.$
\State Set $$\mathcal{L}_j^{(0)} = \left\{ \be_{j}^{(0)} \in \zpsk{K+ \ell} \mid \lweight{\be_{j}^{(0)}} = \frac{v}{2^a} + \sum_{b=0}^{a-1} \frac{\varepsilon_{b}}{2^b} \right\}$$ for all $j \in \{1, \ldots, 2^a\}.$
 \For{ $ i \in \{1, \ldots a-1\}$}
    \For{$j \in \{1, \ldots, 2^{a-i}\}$} 
        \State Compute $\mathcal{L}_j^{(i)} =  \mathcal{L}_{2j-1}^{(i-1)} \bowtie \mathcal{L}_{2j}^{(i-1)}$ using Algorithm \ref{algo:merge}.
    \EndFor
 \EndFor
    \State Compute $\be  \in \mathcal{L}_1^{(a-1)} \bowtie_{\bs,_2\bs_1,\bA,t-v}^{(a)}  \mathcal{L}_2^{(a-1)}.$ 
        \State Return $\bP\be.$
        \State Else start over at step 1.
 \end{algorithmic}
\end{algorithm}

In the following, we compute the asymptotic average cost of the representation technique for one level and two levels.  In fact, in the case of free codes using two levels gives only a small improvement on using one level and for non-free codes using one level is the optimal choice,
as one can see in the comparison in Section \ref{sec:asymptotic}.

\begin{theorem} \label{thm:rep}
\begin{enumerate}
\item Let us  consider the representation technique on one level, here we need to fix the real numbers $V,L,E$, such that  $0 \leq E \leq  (R_I+L) M-V$.
We fix the internal algorithm parameters $v, \ell, \varepsilon_0$, such that
$ \lim\limits_{n \to \infty} \frac{v}{n}=V, \lim\limits_{n \to \infty} \frac{\ell}{n}=L, \lim\limits_{n \to \infty} \frac{\varepsilon_0}{n}=E$.
If $V=0$, then the asymptotic average time complexity is given by $    S(1,T)-S(R_I+L,V)-S(1-R_I-L,T-V).$
If $V>0$, then the asymptotic average time complexity is given by
\begin{gather*} S(1,T)-S(R_I+L,V)-S(1-R_I-L,T-V) \\
 +
 \max \{S(R_I+L, V/2+E), 2S(R_I+L,V/2+E)-(R_I+L-R_1) \}.
 \end{gather*}  
\item If we consider the representation technique on two levels, we fix the real numbers $V,L,E_0, E_1$ and $U$, such that $0 \leq E_1 \leq (R_I+L)M-V, 0 \leq E_0 \leq (R_I+L)M-V/2-E_1$ and $0 \leq U \leq R_I+L-R_1$.
We fix the internal algorithm parameters $\varepsilon_0, \varepsilon_1$ and $u_1$, such that 
$ \lim\limits_{n \to \infty} \frac{\varepsilon_0}{n} = E_0, \lim\limits_{n \to \infty} \frac{\varepsilon_1}{n} = E_1$ and $\lim\limits_{n \to \infty} \frac{u_1}{n} =U,$ which we recall is by Lemma \ref{nrrep} fixed.
Then the asymptotic average time complexity is given by
\begin{gather*}    S(1,T)-S(R_I+L,V)-S(1-R_I-L,T-V) +   \\
\max \{S(R_I+L, V/4+E_0+E_1/2),  S(R_I+L, V/2+E_1)-U, \\
2S(R_I+L, V/4+E_0+E_1/2) -U, \\ 2S(R_I+L,V/2+E_1)-(R_I+L-R_1)-U \}. 
\end{gather*} 
\end{enumerate}
\end{theorem}

\begin{proof}
 \begin{enumerate}
     \item The initial lists   $\mathcal{L}_1^{(0)}$ and $\mathcal{L}_2^{(0)}$ are both of size $L_0 = F_L(K+\ell,v_0, p^s)$, where $v_0 = \frac{v}{2} +\varepsilon$.
     We compute $\mathcal{L}_1^{(1)}= \mathcal{L}_1^{(0)} \bowtie \mathcal{L}_2^{(0)} $  using Algorithm \ref{algo:last-merge} on the following inputs: the positive integers $v,t-v,u=K+\ell-k_1$ and $\bB \in \zpsk{(K+\ell-k_1) \times (K+\ell)}, \bs_2 \in \zpsk{K+\ell-k_1}$. This has the asymptotic cost of \[\max\{S(R_I+L,V/2+E), 2S(R_I+L,V/2+E)-(R_I+L-R_1)\}.\]
     
     Finally, we note that the probability that the error vector has the assumed weight distribution of the representation technique is 
 $$ F_L\left(K+\ell, v,p^s\right)F_L(n-K-\ell,t-v,p^s)F_L(n,t,p^s)^{-1}.$$ Hence, the asymptotic number of iterations is given by \[S(1,T)-S(R_I+L,V)-S(1-R_I-L,T-V).\]
 \item The initial lists $\mathcal{L}_j^{(0)}$ are all of size $L_0 = F_L(K+\ell,v_0, p^s)$, where $v_0 = \frac{v}{4} + \varepsilon_0 +\varepsilon_1/2$.
 On the first level,  the two merges $\mathcal{L}_1^{(1)}= \mathcal{L}_1^{(0)} \bowtie \mathcal{L}_2^{(0)}$ and $\mathcal{L}_2^{(1)}=\mathcal{L}_3^{(0)} \bowtie \mathcal{L}_4^{(0)}$ on $u_1$ positions cost asymptotically
  $$\max\{S(R_I+L, V/4+E_0+E_1/2), 2S(R_I+L, V/4+E_0+E_1/2) -U\}.$$
  The lists  $\mathcal{L}_b^{(1)}$ have asymptotic sizes $ S(R_I+L, V/2+E_1)-U.$
  
 On the second level we compute $\mathcal{L}_1^{(2)}= \mathcal{L}_1^{(1)} \bowtie \mathcal{L}_2^{(1)} $  using Algorithm \ref{algo:last-merge} on the following inputs: the positive integers $v,t-v,u=K+\ell-k_1$ and $\bB \in \zpsk{(K+\ell-k_1) \times (K+\ell)}, \bs_2 \in \zpsk{K+\ell-k_1}$. This has the asymptotic cost of \[\max\{(S(R_I+L,V/2+E_1)-U),2S(R_I+L,V/2+E_1)-(R_I+L-R_1)-U\}.\]
 Note that the addition of $U$ results from the fact, that we already did merge on $u_1$ positions, thus we only need to merge on $K+\ell-k_1-u_1$ positions. 
 The success probability stays the same as in the first level, since it only considers the splitting into the weights $v$ and $t-v.$
 \end{enumerate}

\end{proof}

\subsection{Lee-BJMM Algorithm}
The ISD algorithm BJMM \cite{bjmm} in the Hamming metric is a mixture of the two techniques considered above. We consider the case of two levels and three levels. In both the cases, merging at the first level is done using Algorithm \ref{algo:merge-concat} (similar to Wagner's algorithm), and merging at later levels is done using Algorithm \ref{algo:merge} (similar to the representation technique). Note also, that Lee BJMM at one level is same as Wagner's approach on one level. Note that the optimal choice for the Hamming metric is three levels. For the Lee metric, we observe in Section \ref{sec:asymptotic}, that two levels is the optimal choice.

\begin{corollary}
\begin{enumerate}
\item Let us consider Lee-BJMM on two levels, and fix the real numbers $V,L,E$, and $U$ such that 
 $0 \leq E \leq  (R_I+L) M-V,$ and $0\leq U \leq R_I+L-R_1$.
We fix the internal algorithm parameters $v, \ell, \varepsilon_0$ and $u_1$, such that
$ \lim\limits_{n \to \infty} \frac{v}{n}=V, \lim\limits_{n \to \infty} \frac{\ell}{n}=L, \lim\limits_{n \to \infty} \frac{\varepsilon_0}{n}=E$ and $\lim\limits_{n \to \infty} \frac{u_1}{n}=U$,  which we recall is by Lemma \ref{nrrep} fixed.
The asymptotic average time complexity of BJMM on two levels is given by

\begin{gather*} S(1,T)-S(R_I+L,V)-S(1-R_I-L,T-V) \\
 +
 \max \{S((R_I+L)/2, V/4+ E/2), S(R_I+L,V/2+E)-U, \\  2(S(R_I+L, V/2+E)-U)-(R_I+L-R_1)+U \}. 
 \end{gather*} 
\item Let us consider now Lee-BJMM on three levels, and fix the real numbers $V,L$,  $E_0, E_1$, and $U_0,U_1$ such that  $0 \leq E_1 \leq  (R_I+L) M-V, 0 \leq E_0 \leq (R_I+L)M- V/2-E_1$ and $0\leq U_0 \leq U_1 \leq R_I+L-R_1$.
We fix the internal algorithm parameters $v, \ell, \varepsilon_0, \varepsilon_1, u_0$ and $u_1$, such that
$ \lim\limits_{n \to \infty} \frac{v}{n}=V, \lim\limits_{n \to \infty} \frac{\ell}{n}=L, \lim\limits_{n \to \infty} \frac{\varepsilon_i}{n}=E_i$ and $\lim\limits_{n \to \infty} \frac{u_i}{n}=U_i$ for $i \in \{0,1\}$, which we recall are by Lemma \ref{nrrep} fixed.
The asymptotic average time complexity  of BJMM at three level is given by
\begin{gather*} S(1,T)-S(R_I+L,V)-S(1-R_I-L,T-V) \\
 +
 \max \{S((R_I+L)/2, V/8+E_1/4 + E_0/2), S(R_I+L,V/4+E_1/2+E_0)-U_0, \\ 2S(R_I+L, V/4+E_1/2+E_0)-U_1-U_0, S(R_I+L, V/2 + E_1)-U_1, \\ 2S(R_I+L, V/2+E_1)-(R_I+L-R_1)-U_1 \}.
 \end{gather*} 

\end{enumerate}
\end{corollary}

The proof works similar to Theorem \ref{thm:wagner} and Theorem \ref{thm:rep}, where in the case of two levels, we start with 4 base lists of size $F_L((K+\ell)/2,v/4+\varepsilon_0/2)$ and in the case of three levels we start with 8 base lists of size $F_L((K+\ell)/2,v/8+\varepsilon_1/4+\varepsilon_0/2).$

\section{Comparison}\label{sec:asymptotic}

In this section, we numerically compute the asymptotic complexity of the Lee ISD algorithms proposed in Section \ref{sec:ISD}. For our analysis we will consider codes that achieve the asymptotic Gilbert-Varshamov bound in Theorem \ref{asympt_GV}.
We will exemplary consider $q=7^2$ and assume that $R_1= \lambda R$, for $\lambda \in \{ 0.5, 0.75, 1\}$ and that in general it holds that $$ R \leq \min\left\{1, \lambda R + \frac{(1-\lambda) Rs}{s-1} \right\} \leq R_I \leq \min \left\{(1- \lambda)Rs + \lambda R,1\right\}.$$ 
In \cite[Theorem 22]{modules} it is shown that a code generated by a random matrix is  free and achieves the Gilbert-Varshamov bound with high probability. We will hence assume that $R_I=\min\left\{1, \lambda R + \frac{(1-\lambda) Rs}{s-1} \right\}$  in our asymptotic analysis. Note that for such choices, not all rates $R$ may be achieved.

\begin{figure}
        \centering
        \begin{subfigure}[b]{0.7\textwidth}
            \centering
            \includegraphics[width=\textwidth]{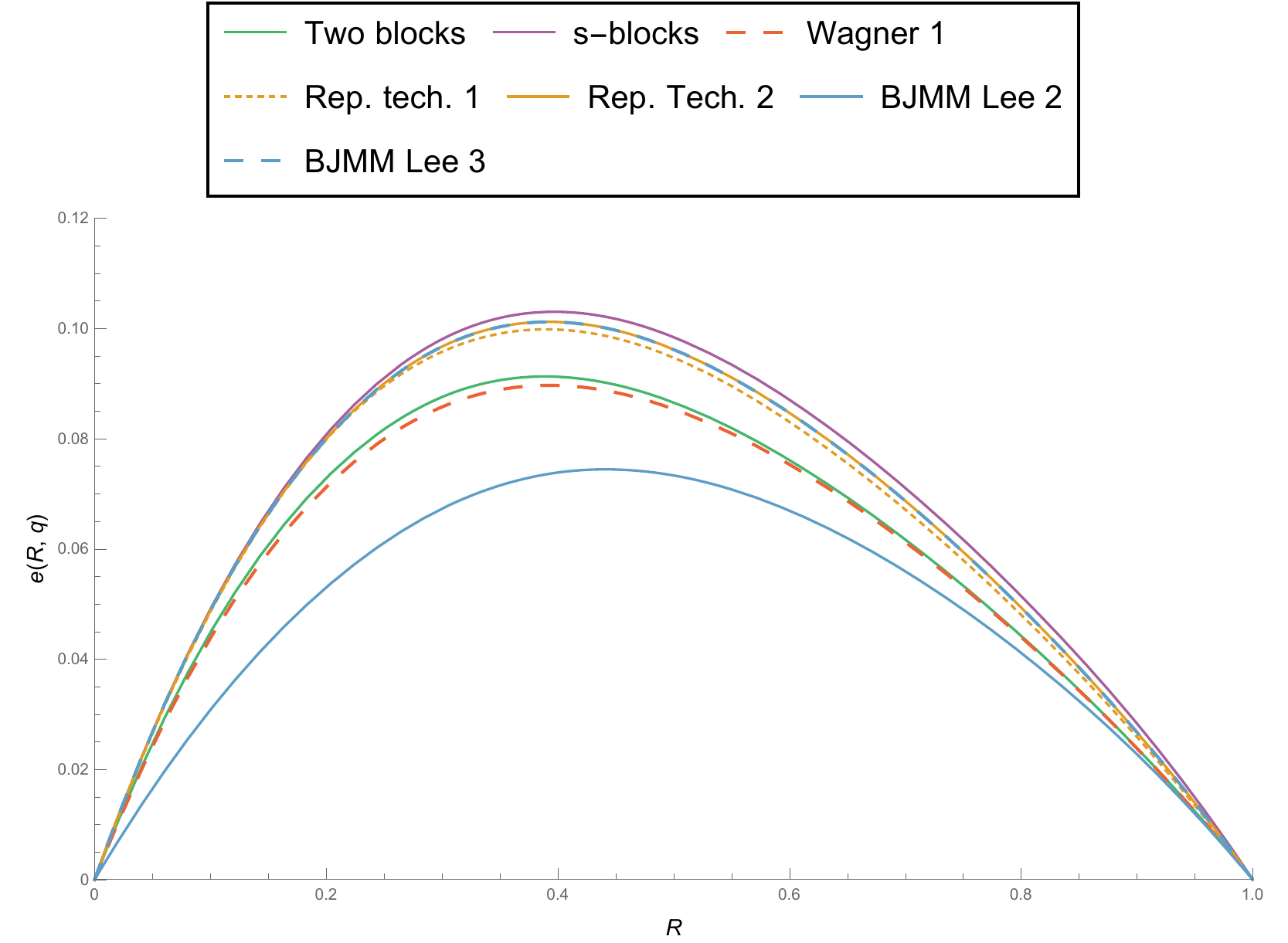}
            \caption{$\lambda = 1$}    
            
        \end{subfigure}
        \hfill
        \begin{subfigure}[b]{0.7\textwidth}  
            \centering 
            \includegraphics[width=\textwidth]{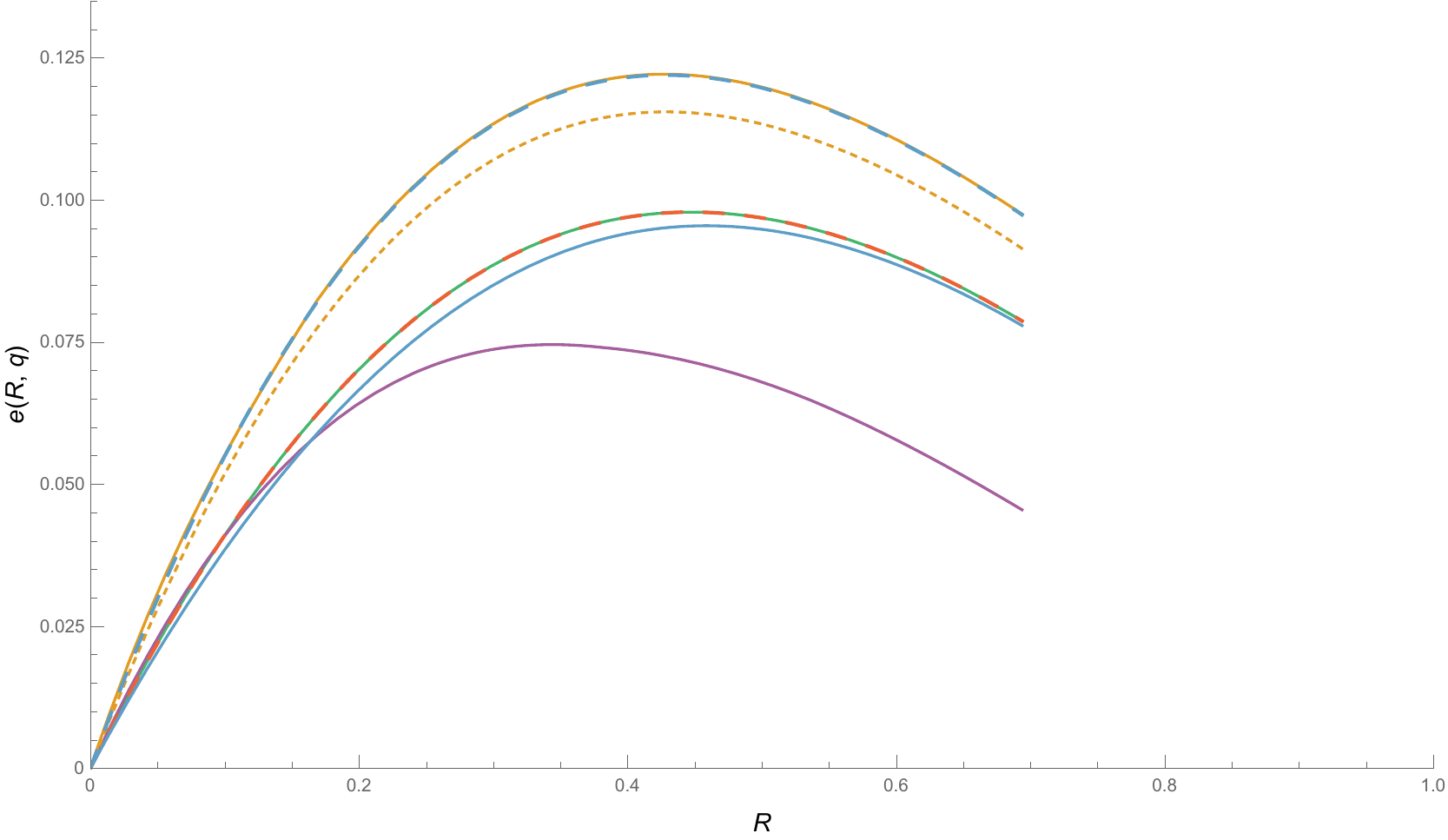}
            \caption{$\lambda = 0.75$}    
            
        \end{subfigure}
        \vskip\baselineskip
        \begin{subfigure}[b]{0.7\textwidth}   
            \centering 
            \includegraphics[width=\textwidth]{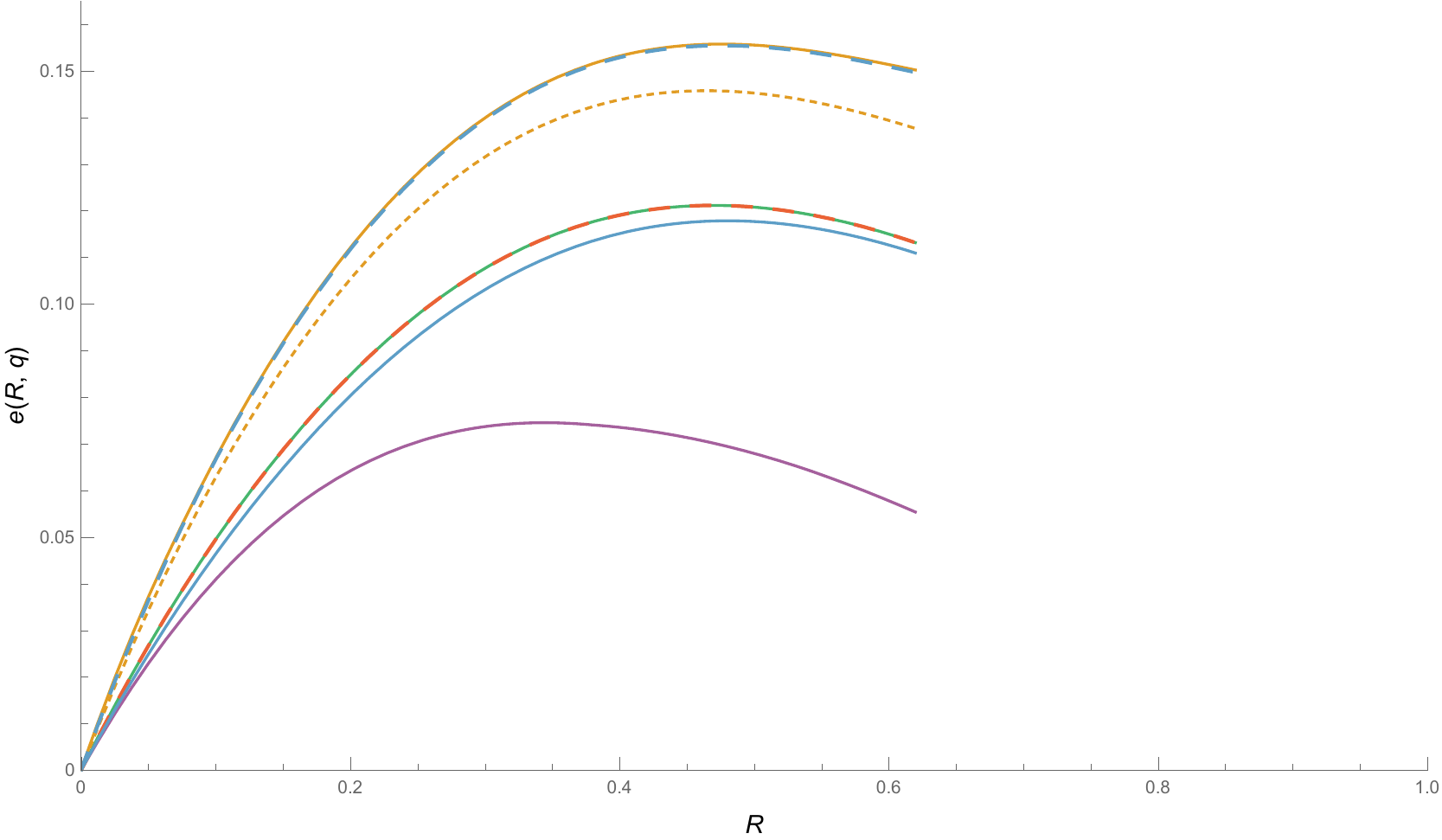}
            \caption{$\lambda = 0.50$}   
            
        \end{subfigure}
        \caption{Comparing asymptotic complexity of different algorithms for different values of $\lambda$. The values are calculated for $q=7^2$.} 
        \label{fig:asymptotics}
    \end{figure}

In Figure \ref{fig:asymptotics}, we compare the asymptotic complexity of all the ISD algorithms at different rates $R$ by optimizing the internal parameters of each algorithm. In Table \ref{table_compare}, we find the complexity in the worst case of each algorithm, i.e., we find $e(R^*,q)$ for $R^* = \argmax\limits_{0 \leq R \leq 1} \left( e(R,q) \right)$. Finally, in Table \ref{table_compareHamm}, we compare the worst-case complexity of Lee ISD algorithms with Hamming ISD algorithms, while fixing the base $q=4$ and $\lambda=1$.

\begin{table}[ht]
 \begin{center}
 \begin{tabular}{|c|c|c | c | c|c|c|}
 \hline 
   &  \multicolumn{2}{|c|}{$\lambda = 1$} & \multicolumn{2}{|c|}{$\lambda = 0.75$} & \multicolumn{2}{|c|}{$\lambda = 0.5$} \\ \hline 
   & $R^*$ & $e(R^*,q)$ & $R^*$ & $e(R^*,q)$ & $R^*$ & $e(R^*,q)$ \\\hline
   Two-Blocks & 0.3886 & 0.0913 & 0.4473 & 0.0978 & 0.4694  & 0.1211  \\
   $s$-Blocks & 0.3969 & 0.1030 & 0.3441 & 0.0745 & 0.3441  & 0.07453 \\
   Wagner $a=1$ & 0.3925 & 0.0897 & 0.4473 & 0.0978 & 0.4694  & 0.1211 \\
   Wagner $a=2$ & 0.3925 & 0.0897 & 0.4473 & 0.0978 & 0.4694  & 0.1211 \\
   Rep. tech. $a=1$ & 0.3896 & 0.0998 & 0.4288 & 0.1155 & 0.4648 & 0.1457 \\
   Rep. tech. $a=2$ & 0.3922 & 0.1012 & 0.4275 & 0.1221 & 0.4757 & 0.1557\\ 
   BJMM level 2 & 0.4414 & 0.07440 & 0.4587 & 0.0954 & 0.4801 & 0.1178 \\ 
   BJMM level 3 & 0.3921 & 0.1012 & 0.4282 & 0.1220 & 0.4754 & 0.1554
 
 \\ \hline
  \end{tabular}
\end{center}  
  \caption{Comparison of the asymptotic complexity of all the algorithms at rate $R^* = \argmax\limits_{0 \leq R \leq 1} \left( e(R,q) \right)$. The values are calculated for $q=7^2$.}\label{table_compare}
\end{table}

\begin{table}[ht]
 \begin{center}
 \begin{tabular}{|c|c|}
 \hline 
   & $e(R^*,q)$ \\\hline
   \textbf{Lee Metric} & \\
   Prange & 0.0575 \\
   $s$-Blocks & 0.0575 \\
   Two-Blocks & 0.0556 \\
   Wagner $a=1$ &  0.0556 \\
   Rep. tech. $a=1$ & 0.0569 \\
   Rep. tech. $a=2$ & 0.0571 \\ 
   BJMM level 2 & 0.05265 \\
 BJMM level 3&  0.0557 \\ \hline
 \textbf{Hamming Metric} & \\
 BJMM-MO & 0.04294 \\
 Stern & 0.04987 \\
 Prange & 0.05095  \\ \hline
  \end{tabular}
\end{center}  
  \caption{Comparison with Hamming metric for $q=4$ and $\lambda=1$. The values for Hamming metric ISD algorithms BJMM-MO and Stern are from \cite[Table 3]{klamti} and \cite[Table 1]{hirose}, respectively.}\label{table_compareHamm}
\end{table}

For free codes, we can see from Figure \ref{fig:asymptotics} (A) that BJMM on two levels is the fastest algorithm. In addition, we have observed that Wagner's approach on one level and the two-blocks algorithms have a similar cost, outperforming the representation technique algorithms.

The situation changes drastically when we consider non-free codes. In this case, the $s$-blocks algorithm outperforms all other algorithms, even BJMM.

 \begin{remark}
 Asymptotically, the algorithm based on Wagner's approach on two levels results in the same cost as the level one variant. Indeed it is easy to see that the term $U$ in the asymptotic complexity either gets canceled out or the cost minimizes at $U=0$. This is also observed in the values in Table \ref{table_compare}. 
 \end{remark}
 
 If we compare the asymptotic cost of the Lee-metric ISD algorithms to the corresponding algorithms in the Hamming metric, namely BJMM-MO, Stern and Prange we observe that decoding a random code in the Lee metric has a larger cost than in the Hamming metric, which makes it a promising alternative for cryptographic applications.

We provide all the SAGE \cite{sage} and Mathematica \cite{Mathematica} programs which can be used to compute the finite and asymptotic workfactors, respectively, at  \url{https://git.math.uzh.ch/isd/lee-isd/lee-isd-algorithm-complexities.git}

\section{Conclusion}\label{sec:concl}
In this paper we analyzed the hardness of the syndrome decoding problem in the  Lee metric. 

We extended the reduction  of Barg \cite{barg1994some}  to the syndrome decoding and the given weight codeword problem of additive weights, thus including also the NP-completeness of the syndrome decoding problem in the Lee metric.  In order to analyze more in depth the hardness of this problem, we provided several decoding algorithms for random linear Lee-metric codes and provided their asymptotic cost. 

The NP-completeness and the observation that, for a fixed set of code parameters, the cost of decoding a random code in the Lee metric is larger than in the Hamming metric, make the Lee metric a promising candidate for code-based cryptography.  We thus emphasize the need of further research on Lee-metric codes with suitable parameters for applications in cryptography. For example, Low-Lee-Density Parity-Check codes were recently introduced \cite{SantiniLowLee}.


\section*{Acknowledgments}\label{sec:ack}
 The first author  is  supported by the Swiss National Science Foundation grant number 195290. The second author  is  supported by the Estonian Research Council grant number PRG49.

\bibliographystyle{plain}

\bibliography{References}

\begin{thebibliography}{10}

\bibitem{Abramson1976}
M.~Abramson.
\newblock Restricted combinations and compositions.
\newblock {\em Fibonacci Quart.}, pages 439--452, 1976.

\bibitem{hastola}
H.~Astola and I.~Tabus.
\newblock Bounds on the size of {L}ee-codes.
\newblock In {\em 2013 8th International Symposium on Image and Signal
  Processing and Analysis (ISPA)}, pages 471--476. IEEE, 2013.

\bibitem{leeas}
J.~Astola.
\newblock On the asymptotic behaviour of {L}ee-codes.
\newblock {\em Discrete applied mathematics}, 8(1):13--23, 1984.

\bibitem{ISDreview}
M.~Baldi, A.~Barenghi, F.~Chiaraluce, G.~Pelosi, and P.~Santini.
\newblock A finite regime analysis of information set decoding algorithms.
\newblock {\em Algorithms}, 12(10), 2019.

\bibitem{barg1994some}
S.~Barg.
\newblock Some new {NP}-complete coding problems.
\newblock {\em Problemy Peredachi Informatsii}, 30(3):23--28, 1994.

\bibitem{Becker2012}
A.~Becker, A.~Joux, A.~May, and A.~Meurer.
\newblock Decoding random binary linear codes in $2^{n/20}$: How 1 + 1 = 0
  improves information set decoding.
\newblock In D.~Pointcheval and T.~Johansson, editors, {\em Advances in
  Cryptology - {EUROCRYPT} 2012}, volume 7237 of {\em Lecture Notes in Computer
  Science}, pages 520--536. Springer Verlag, 2012.

\bibitem{bjmm}
A.~Becker, A.~Joux, A.~May, and A.~Meurer.
\newblock Decoding random binary linear codes in {$2^{n/20}$}: How {$1+ 1= 0$}
  improves information set decoding.
\newblock In {\em Annual international conference on the theory and
  applications of cryptographic techniques}, pages 520--536. Springer, 2012.

\bibitem{berlekamp1968algebraic}
E.~Berlekamp.
\newblock {\em Algebraic coding theory}.
\newblock World Scientific, 1968.

\bibitem{Berlekamp1978}
E.~Berlekamp, R.~{McE}liece, and H.~van Tilborg.
\newblock On the inherent intractability of certain coding problems.
\newblock {\em IEEE Trans. on Inf. Theory}, 24(3):384--386, May 1978.

\bibitem{bernstein2011smaller}
D.~J. Bernstein, T.~Lange, and C.~Peters.
\newblock Smaller decoding exponents: ball-collision decoding.
\newblock In {\em Annual Cryptology Conference}, pages 743--760. Springer,
  2011.

\bibitem{modules}
E.~Byrne, A.-L. Horlemann, K.~Khathuria, and V.~Weger.
\newblock Density of free modules over finite chain rings.
\newblock {\em arXiv preprint arXiv:2106.09403}, 2021.

\bibitem{canteaut1998new}
A.~Canteaut and F.~Chabaud.
\newblock A new algorithm for finding minimum-weight words in a linear code:
  application to {M}c{E}liece's cryptosystem and to narrow-sense {BCH} codes of
  length 511.
\newblock {\em IEEE Trans. on Inf. Theory}, 44(1):367--378, 1998.

\bibitem{canteautsendrier}
A.~Canteaut and N.~Sendrier.
\newblock Cryptanalysis of the original {McE}liece cryptosystem.
\newblock In {\em International Conference on the Theory and Application of
  Cryptology and Information Security}, pages 187--199. Springer, 1998.

\bibitem{chabaud}
F.~Chabaud.
\newblock Asymptotic analysis of probabilistic algorithms for finding short
  codewords.
\newblock In {\em Eurocode’92}, pages 175--183. Springer, 1993.

\bibitem{debris}
A.~Chailloux, T.~Debris-Alazard, and S.~Etinski.
\newblock Classical and quantum algorithms for generic syndrome decoding
  problems and applications to the {L}ee metric.
\newblock {\em arXiv preprint arXiv:2104.12810}, 2021.

\bibitem{NISTreport2016}
L.~Chen, Y.-K. Liu, S.~Jordan, D.~Moody, R.~Peralta, R.~Perlner, and
  D.~Smith-Tone.
\newblock Report on post-quantum cryptography.
\newblock Technical Report NISTIR 8105, National Institute of Standards and
  Technology, 2016.

\bibitem{fiat}
A.~Fiat and A.~Shamir.
\newblock How to prove yourself: Practical solutions to identification and
  signature problems.
\newblock In {\em Conference on the theory and application of cryptographic
  techniques}, pages 186--194. Springer, 1986.

\bibitem{finiasz}
M.~Finiasz and N.~Sendrier.
\newblock Security bounds for the design of code-based cryptosystems.
\newblock In {\em International Conference on the Theory and Application of
  Cryptology and Information Security}, pages 88--105. Springer, 2009.

\bibitem{gaborit2016hardness}
P.~Gaborit and G.~Z{\'e}mor.
\newblock On the hardness of the decoding and the minimum distance problems for
  rank codes.
\newblock {\em IEEE Trans. on Inf. Theory}, 62(12):7245--7252, 2016.

\bibitem{saddle}
D.~Gardy and P.~Sol{\'e}.
\newblock Saddle point techniques in asymptotic coding theory.
\newblock In {\em Workshop on Algebraic Coding}, pages 75--81. Springer, 1991.

\bibitem{klamti}
C.~T. Gueye, J.~B. Klamti, and S.~Hirose.
\newblock Generalization of {BJMM-ISD} using {M}ay-{O}zerov nearest neighbor
  algorithm over an arbitrary finite field $\mathbb{F}_q$.
\newblock In {\em Codes, Cryptology and Information Security}, pages 96--109.
  Springer International Publishing, 2017.

\bibitem{hirose}
S.~Hirose.
\newblock {M}ay-{O}zerov algorithm for nearest-neighbor problem over
  $\mathbb{F}_q$ and its application to information set decoding.
\newblock In {\em International Conference for Information Technology and
  Communications}, pages 115--126. Springer, 2016.

\bibitem{Horlemann2019}
A.-L. Horlemann-Trautmann and V.~Weger.
\newblock Information set decoding in the {L}ee metric with applications to
  cryptography.
\newblock {\em Advances in Mathematics of Communications}, online, 2019.

\bibitem{Mathematica}
Wolfram~Research{,} Inc.
\newblock Mathematica, {V}ersion 12.3.1.
\newblock Champaign, IL, 2021.

\bibitem{interlando2018generalization}
C.~Interlando, K.~Khathuria, N.~Rohrer, J.~Rosenthal, and V.~Weger.
\newblock Generalization of the ball-collision algorithm.
\newblock {\em Journal of Algebra Combinatorics Discrete Structures and
  Applications}, 7(2):195--207, 2018.

\bibitem{Lee1958}
C.~{Lee}.
\newblock Some properties of nonbinary error-correcting codes.
\newblock {\em IRE Trans. Inf. Theory}, 4(2):77--82, Jun. 1958.

\bibitem{Lee1988}
P.~Lee and E.~Brickell.
\newblock An observation on the security of {M}c{E}liece's public-key
  cryptosystem.
\newblock In {\em Advances in Cryptology - EUROCRYPT 88}, pages 275--280.
  Springer Verlag, 1988.

\bibitem{Leon1988}
J.~S. Leon.
\newblock A probabilistic algorithm for computing minimum weights of large
  error-correcting codes.
\newblock {\em IEEE Trans. on Inf. Theory}, 34(5):1354--1359, September 1988.

\bibitem{mmt}
A.~May, A.~Meurer, and E.~Thomae.
\newblock Decoding random linear codes in {$\tilde{\mathcal{O}}(2^{0.054 n})$}.
\newblock In {\em International Conference on the Theory and Application of
  Cryptology and Information Security}, pages 107--124. Springer, 2011.

\bibitem{McEliece1978}
R.~McEliece.
\newblock A public-key cryptosystem based on algebraic coding theory.
\newblock {\em DSN Progress Report}, pages 114--116, 1978.

\bibitem{meurer2013coding}
A.~Meurer.
\newblock {\em A coding-theoretic approach to cryptanalysis}.
\newblock PhD thesis, Ruhr Universit{\"a}t Bochum, 2013.

\bibitem{niebuhr}
R.~Niebuhr, E.~Persichetti, P.-L. Cayrel, S.~Bulygin, and J.~Buchmann.
\newblock On lower bounds for information set decoding over $\mathbb{F}_q$ and
  on the effect of partial knowledge.
\newblock {\em Int. J. Inf. Coding Theory}, 4(1):47--78, 2017.

\bibitem{Niederreiter1986}
H.~Niederreiter.
\newblock Knapsack-type cryptosystems and algebraic coding theory.
\newblock {\em Probl. Contr. and Inform. Theory}, 15:159--166, 1986.

\bibitem{peters}
C.~Peters.
\newblock Information-set decoding for linear codes over $\mathbb{F}_q$.
\newblock In {\em International Workshop on Post-Quantum Cryptography}, pages
  81--94. Springer, 2010.

\bibitem{Prange1962}
E.~Prange.
\newblock The use of information sets in decoding cyclic codes.
\newblock {\em IRE Trans. Inf. Theory}, 8(5):5--9, Sep. 1962.

\bibitem{sumrank}
S.~Puchinger, J.~Renner, and J.~Rosenkilde.
\newblock Generic decoding in the sum-rank metric.
\newblock In {\em 2020 IEEE International Symposium on Information Theory
  (ISIT)}, pages 54--59. IEEE, 2020.

\bibitem{SantiniLowLee}
P.~{Santini}, M.~{Battaglioni}, F.~{Chiaraluce}, M.~{Baldi}, and
  E.~{Persichetti}.
\newblock Low-{L}ee-density parity-check codes.
\newblock In {\em ICC 2020 - 2020 IEEE International Conference on
  Communications (ICC)}, pages 1--6, 2020.

\bibitem{stern}
J.~Stern.
\newblock A method for finding codewords of small weight.
\newblock In {\em International Colloquium on Coding Theory and Applications},
  pages 106--113. Springer, 1988.

\bibitem{Stern1994}
J.~Stern.
\newblock A new identification scheme based on syndrome decoding.
\newblock In Douglas~R. Stinson, editor, {\em Advances in Cryptology ---
  CRYPTO' 93}, pages 13--21, Berlin, Heidelberg, 1994. Springer Berlin
  Heidelberg.

\bibitem{sage}
{The Sage Developers}.
\newblock {\em {S}ageMath, the {S}age {M}athematics {S}oftware {S}ystem
  ({V}ersion 8.4)}, 2018.
\newblock {\tt https://www.sagemath.org}.

\bibitem{Ulrich1957}
W.~{Ulrich}.
\newblock Non-binary error correction codes.
\newblock {\em The Bell System Technical Journal}, 36(6):1341--1388, Nov. 1957.

\bibitem{wagner}
D.~Wagner.
\newblock A generalized birthday problem.
\newblock In {\em Annual International Cryptology Conference}, pages 288--304.
  Springer, 2002.

\end{thebibliography}

\section*{Appendix}

For completeness we include here the proof of Proposition \ref{prop:AW-SDP}.

\begin{proof}
We prove the NP-completeness by a reduction from the 3DM problem. For this, we start with a random instance of 3DM with $T$ of size $t$, and $U \subseteq T \times T \times T$ of size $u$. Let us denote the elements in $T= \{b_1, \ldots, b_t\}$ and in $U= \{\ba_1, \ldots, \ba_u\}$. From this we build the matrix  $\bH^\top \in R^{u \times 3t}$ like $\overline{\bH}^\top$ in the proof of Proposition \ref{prop:GAWCP}.

In Figure \ref{fig:3dm} we provide a toy example to clarify how the construction in our reduction works.
Notice that the very same construction is used also in Proposition \ref{prop:GAWCP}, where we generalize the canonical problem of finding codewords with given weight, and prove it to be NP-complete.

\begin{figure}[ht]
\begin{subfigure}{0.3\textwidth}
\centering
  \begin{tikzpicture}[xscale=1.15]




\draw[fill = mygray] (-1.1, -3.28) rectangle (1.1, -3.72);
\node at (0,-7*0.5) {$\mathbf a_7 = \left\{A, B, C\right\}$};

\draw[fill = mygray] (-1.1, -1.78) rectangle (1.1, -2.22);
\node at (0,-4*0.5) {$\mathbf a_4 = \left\{B, C, D\right\}$};

\draw[fill = mygray] (-1.1, -1.28) rectangle (1.1, -1.72);
\node at (0,-3*0.5) {$\mathbf a_3 = \left\{D, A, B\right\}$};

\draw[fill = mygray] (-1.1, -2.28) rectangle (1.1, -2.72);
\node at (0,-5*0.5) {$\mathbf a_5 = \left\{C,D, A\right\}$};

\node at (0,-6*0.5) {$\mathbf a_6 = \left\{A,D, A\right\}$};
\node at (0,-2*0.5) {$\mathbf a_2 = \left\{C,B, A\right\}$};
\node at (0,-1*0.5) {$\mathbf a_1 = \left\{D,A, B\right\}$};

\end{tikzpicture}
  \caption{ }
\end{subfigure}
\begin{subfigure}{0.69\textwidth}
  \centering
\begin{tikzpicture}[xscale=1.85,yscale=1.32]

\draw[fill =  mygray] (-1.15, -0.1) rectangle (2, 0.15);
\draw[fill =  mygray] (-1.15, -0.1+0.32) rectangle (2, 0.15+0.32);
\draw[fill =  mygray] (-1.15, -0.1 - 0.32) rectangle (2, 0.15 - 0.32);

\draw[fill =  mygray] (-1.16, -0.1 - 0.96) rectangle (2, 0.15 - 0.97);

\node at (0.2,0) {$\mathbf H^\top =  \begin{pmatrix}
\left.\begin{matrix}
0 & 0 & 0 & 1 \\
0 & 0 & 1 & 0 \\
0 & 0 & 0 & 1 \\
0 & 1 & 0 & 0 \\
0 & 0 & 1 & 0 \\
1 & 0 & 0 & 0 \\
1 & 0 & 0 & 0
\end{matrix}\right|
\left.\begin{matrix}
1 & 0 & 0 & 0 \\
0 & 1 & 0 & 0 \\
1 & 0 & 0 & 0 \\
0 & 0 & 1 & 0 \\
0 & 0 & 0 & 1 \\
0 & 0 & 0 & 1 \\
0 & 1 & 0 & 0
\end{matrix}\right|
\begin{matrix}
0 & 1 & 0 & 0 \\
1 & 0 & 0 & 0 \\
0 & 1 & 0 & 0 \\
0 & 0 & 0 & 1 \\
1 & 0 & 0 & 0 \\
1 & 0 & 0 & 0 \\
0 & 0 & 1 & 0
\end{matrix}
\end{pmatrix}$};

\end{tikzpicture}

  \caption{ }
\end{subfigure}
\caption{Example on how to build an $(R,\rm wt)$-SDP instance from a 3DM instance, as in Proposition \ref{prop:AW-SDP}. 
In this example, we have $t = 4$, $u = 7$ and $T = \left\{A, B, C, D\right\}$. 
The triples in the set $U = \left\{\mathbf a_1, \cdots, \mathbf a_7\right\}$ are shown in (a),  while (b) reports the construction of $\mathbf H^\top$ from $U$. 
In (a) we highlighted the subset $W = \{\mathbf a_3, \mathbf a_4, \mathbf a_5, \mathbf a_7\}$, which constitutes a solution to the 3DM instance.
In (b) we highlighted the rows corresponding to the triples in $W$; it is easily seen that summing the rows of $\mathbf H^\top$ corresponding to the triples in $W$ results in the all-one vector, i.e., the target syndrome.}
\label{fig:3dm}
\end{figure}

With this construction, we have that each column of $\bH$ corresponds to an element in $U$, and has weight $3$. 
Let us set the syndrome $\bs$ as the all-one vector of length $3t$.
Assume that we can solve the $(R,\rm wt)$-SDP on the instances $\bH, \bs$ and $t$ in polynomial time. 
Let us consider two cases.

\underline{Case 1:} 
First, assume that the $(R,\rm wt)$-SDP solver returns as answer `yes', i.e., there exists an $\be \in R^u$, of weight less than or equal to $t$ and such that $\be\bH^\top =\bs.$ 
\begin{itemize}
    \item 
We first observe that we must have $\weight{\be}=\card{\supp{\be}}=t$. For this note that each column of $\bH$ adds at most 3 non-zero entries to $\bs$. Therefore, we need to add at least $t$ columns to get $\bs$, i.e., $\card{\supp{\be}} \geq t$ and hence $\weight{\be} \geq t$. As we also have $\weight{\be}\leq t$ by hypothesis, this implies that $\weight{\be}=\card{\supp{\be}}=t$. Moreover, we have that $\weight{e_i}=1$ for all $i \in \supp{\be}$.
\item Secondly, we observe that the weight $t$ solution must be a vector with entries in $\{0,1\}$. For this we note that the matrix $\bH$ has entries in $\{0,1\}$ and has constant column weight three, and since $\card{\supp{\be}}=t$, the supports of the $t$ columns of $\bH$ that sum up to the all-one vector have to be disjoint.  Therefore, we get that the $j$-th equation from the system of equations $ \be \bH^\top =\bs$ is of the form $e_{i} h_{i,j} = 1$ for some $i \in \supp{\be}$. Since $h_{i,j}=1$, we have $e_{i} =1$.
\end{itemize}
Recall from above that the columns of $\bH$ correspond to the elements of $U$. The $t$ columns corresponding to the support of $\be$ are now a solution $W$ to the 3DM problem. This follows from the fact that the $t$ columns have disjoint supports and add up to the all-one vector, which implies that each element of $T$ appears exactly once in each coordinate of the elements of $W$.

\underline{Case 2:} 
If the solver returns as answer `no', this is also the correct answer for the 3DM problem. In fact, it is easy to see that the above construction also associates any solution $W$ of the 3DM to a solution $\be$ of the corresponding $(R, \rm wt)$-SDP.

Thus, if such a polynomial time solver exists, we can also solve the 3DM problem in polynomial time.   
\end{proof}

\end{document}